\documentclass[11pt]{article}

\usepackage[T1]{fontenc}
\usepackage[utf8]{inputenc}
\usepackage{lmodern}
\usepackage{microtype}

\usepackage[a4paper,margin=1in,headheight=14pt,headsep=12pt]{geometry}
\usepackage{setspace}
\usepackage{amsmath,amssymb,amsthm,mathtools}

\usepackage{graphicx}
\usepackage{booktabs}
\usepackage{array}
\usepackage{subcaption}
\usepackage{caption}
\usepackage[round,authoryear,sort&compress]{natbib}
\usepackage{xcolor}
\definecolor{linkblue}{HTML}{1F3A93}
\usepackage[
  colorlinks=true,
  citecolor=linkblue,
  linkcolor=linkblue,
  urlcolor=linkblue,
  bookmarksopen=true,
  bookmarksnumbered=true,
  pdftitle={Multiplicative Contractions, Additive Recoveries: Functional-Form Restrictions on Risk Exposure Dynamics},
  pdfauthor={Liang Chen},
  pdfkeywords={intermediary asset pricing, VaR constraints, drawdown-recovery asymmetry, regime-conditional dynamics, FINRA margin debt}
]{hyperref}
\usepackage[capitalize,noabbrev]{cleveref}

\usepackage{fancyhdr}
\fancypagestyle{plain}{%
  \fancyhf{}\fancyfoot[C]{\small\thepage}%
  }

\usepackage{titlesec}
\titleformat{\section}{\normalfont\Large\bfseries}{\thesection}{0.6em}{}
\titleformat{\subsection}{\normalfont\large\bfseries}{\thesubsection}{0.55em}{}
\titleformat{\subsubsection}{\normalfont\normalsize\bfseries}{\thesubsubsection}{0.5em}{}
\titleformat{\paragraph}[runin]{\normalfont\normalsize\bfseries}{}{0pt}{}
\titlespacing*{\section}{0pt}{14pt plus 2pt}{6pt}
\titlespacing*{\subsection}{0pt}{10pt plus 2pt}{4pt}
\titlespacing*{\subsubsection}{0pt}{8pt plus 2pt}{2pt}
\titlespacing*{\paragraph}{0pt}{6pt plus 1pt}{0.7em}

\newcommand{\E}{\mathbb{E}}
\newcommand{\Var}{\mathrm{Var}}
\newcommand{\Cov}{\mathrm{Cov}}

\DeclareMathOperator*{\argmin}{arg\,min}

\newcommand{\agents}{N}                 
\newcommand{\exposure}[1]{x_{#1}}       
\newcommand{\totalalloc}{X}             
\newcommand{\capital}[1]{R_{#1}}        
\newcommand{\vark}{k}                   
\newcommand{\flow}{\phi}                

\newcommand{\ai}{\alpha}                
\newcommand{\md}{\beta}                 

\newcommand{\price}{P}                  
\newcommand{\vol}{\sigma}               
\newcommand{\impact}{\lambda}           
\newcommand{\stressth}{\vol^{(q)}}      

\newcommand{\retfactor}{\hat{\rho}}     

\theoremstyle{plain}
\newtheorem{theorem}{Theorem}
\newtheorem{proposition}[theorem]{Proposition}

\theoremstyle{definition}

\newtheorem{assumption}{Assumption}
\theoremstyle{remark}

\title{%
  \vspace{-1.2em}%
  {\LARGE\bfseries Multiplicative Contractions, Additive Recoveries}\\[6pt]
  {\large\mdseries Functional-Form Restrictions on Risk Exposure Dynamics}%
  \vspace{-0.4em}%
}

\author{Liang Chen}

\date{\small \today}

\begin{document}

\maketitle
\thispagestyle{plain}

\begin{abstract}
We test a regime-conditional functional-form restriction on
aggregate risk-exposure dynamics implied by VaR-constrained
intermediary models: exposures should contract
\emph{multiplicatively} when capital constraints bind and
grow approximately \emph{additively} (level-independent)
when constraints are slack.  The contraction half follows
mechanically from binding VaR constraints, as in
\citet{brunnermeier2009market}, \citet{adrian2010liquidity},
and \citet{he2013intermediary}.  The additive-rebuild
prediction does not appear formally in this literature; we
derive it under a constant-rate capital replenishment
assumption and test the joint restriction on FINRA monthly
margin debt (1997--2026).

We report two empirical findings.  First, a regime-interacted
regression of detrended margin growth on the lagged level
($T = 350$ months) yields a calm-period slope of $-0.040$
($p = 0.082$, consistent with additive growth) and an
implied stress-period slope of $-0.205$ ($p < 0.001$,
consistent with multiplicative contraction); the Wald test
on the regime $\times$ level interaction rejects equal level
dependence at the $0.16\%$ level.  Second, the same
restriction implies a price-level prediction that the
drawdown-recovery duration ratio should increase with
crash depth.  On 73 S\&P~500 episodes (1950--2026) a Cox
proportional-hazard model yields a depth coefficient of
$-13.75$ ($p < 10^{-7}$), implying a $75\%$ reduction in
per-period recovery hazard per 10-percentage-point increase
in drawdown depth; a continuous-depth regression gives
$\hat\beta = 1.22$ ($p = 0.047$; $\hat\beta = 1.59$,
$p < 0.001$ excluding the 1980--82 Volcker episode); the
median duration ratio for crashes exceeding 30\% is
$3.1\times$, and the pattern replicates across eight
additional international equity indices.  Calibrated Heston
stochastic-volatility, Markov regime-switching, and block
bootstrap nulls match the price-level duration asymmetry
but contain no exposure state variable, so the
regime-conditional flip on direct exposures is the
restriction they cannot speak to.

We do not claim that the exposure test resolves identification
of the intermediary mechanism: FINRA margin debt is a noisy
proxy that mixes leveraged demand, collateral revaluation, and
sample-composition shifts.  We claim only that the
regime-conditional functional form is a sharper empirical
target than return-level moments alone, and that confirming it
on margin debt is consistent with---not proof of---the
constrained-intermediary mechanism.  A companion test on CFTC
weekly speculative positioning is left for future work; the
data assembly and methodology are documented in
\Cref{sec:cot,sec:appendix_replication}.
\end{abstract}

\section{Introduction}
\label{sec:introduction}

Equity drawdowns unfold systematically faster than recoveries.
In the S\&P~500 since 1950, six episodes have exceeded a 30\%
peak-to-trough decline; the median time from peak to trough is
362 trading days, while the median recovery to the prior peak
is 1{,}123 days---a duration ratio of $3.1\times$.  The
pattern persists across markets and institutional regimes
(\Cref{sec:cross_market}) and is the price-level counterpart
of the return-asymmetry regularities documented since
\citet{black1976studies, christie1982stochastic,
nelson1991conditional}.

The mechanism most often invoked to rationalize this pattern is
the binding capital constraint of risk-bearing intermediaries.
Under a value-at-risk or margin constraint of the form $\vark
\cdot \vol_t \cdot \exposure{i,t} \leq \capital{i,t}$, a
volatility spike or capital impairment forces intermediary~$i$
to reduce exposure proportionally to its current size.  This
mechanically yields a multiplicative contraction
$\exposure{i,t+1} = \md_{i,t} \cdot \exposure{i,t}$ with
contraction factor $\md_{i,t} =
(\capital{i,t+1}/\capital{i,t}) \cdot (\vol_t/\vol_{t+1})$,
typically less than unity in stress because volatility rises
and capital is impaired.  This is the building block of the
liquidity spirals of \citet{brunnermeier2009market}, the
procyclical leverage of \citet{adrian2010liquidity}, and the
intermediary asset pricing of \citet{he2013intermediary}.

What this literature does not formally pin down is what happens
once the constraint relaxes.  The contraction half of the
adjustment is mechanical and well-understood; the rebuild half
is typically left implicit or absorbed into a generic
GARCH-style autoregression on returns.  This is not a minor
omission.  Suppose a contraction reduces aggregate exposure to
$\md \totalalloc$ with $\md \in (0,1)$.  Recovery to the prior
peak under \emph{additive} growth at rate $\ai$ takes
$T_{\text{rec}} = (1 - \md)\totalalloc / \ai$ periods, which
scales \emph{linearly} in the contraction depth $1 - \md$.
Under \emph{multiplicative} growth at rate $g$, the same
recovery takes $T_{\text{rec}} = -\log(\md) / \log(1+g)$
periods, which scales \emph{logarithmically} in $1 - \md$.
A model with multiplicative recovery therefore predicts
duration-depth scaling that is concave and bounded; a model
with additive recovery predicts linear scaling without bound.
The observed pattern that severe-crash recoveries take three
to five times as long as the contractions themselves
(\Cref{sec:dd_recovery}) is more consistent with the linear
than with the logarithmic scaling, motivating the additive
specification on the rebuild side that we develop and test
below.

\paragraph{This paper.}
We close this gap with three contributions.  First, we show
that under a single additional assumption---approximately
constant-rate capital replenishment during calm periods (from
fund inflows, retained earnings, and gradual risk-budget
recovery)---the slack-regime dynamics of the same
VaR-constrained intermediary model imply a sharp asymmetric
counterpart: exposures grow at a level-independent
\emph{additive} rate $\Delta \exposure{i,t} = \ai_i$.  The
full functional-form restriction is therefore a regime-conditional
property of intermediary equilibrium:
\begin{align*}
  \text{Stress } (\vol_t, \vol_{t+1} > \stressth):
    \quad &\exposure{i,t+1} \;=\; \md_{i,t} \cdot \exposure{i,t},
    \quad \md_{i,t} \in (0, 1), \\
  \text{Calm } (\vol_t, \vol_{t+1} \leq \stressth):
    \quad &\E[\exposure{i,t+1} - \exposure{i,t}] \;=\; \ai_i,
    \quad \Cov(\Delta \exposure{i,t},
    \exposure{i,t}) = 0.
\end{align*}
We refer to this as the ``multiplicative-contraction,
additive-replenishment'' (MC--AR) restriction.

Second, we test the restriction on aggregate exposure data.
On FINRA monthly margin debt (1997--2026, $T = 350$), a
regime-interacted regression of detrended margin growth on
the lagged margin level yields a calm-period slope of
$-0.040$ ($p = 0.082$, consistent with R2's
level-independence prediction) and an implied stress-period
slope of $-0.205$ ($p < 0.001$, consistent with R1's
proportional contraction).  The Wald test on the regime
$\times$ level interaction coefficient ($\hat b_S = -0.165$,
HAC SE $0.052$) rejects equal level dependence across
regimes at the $0.16\%$ level; the stress-regime level
dependence is approximately $5\times$ the calm-regime
estimate.  We treat this as a sharper empirical target than
return-level moments rather than as identification of the
underlying mechanism: margin debt is a noisy proxy for
intermediary balance-sheet exposure (it mixes leveraged
client demand, collateral revaluation, and sample-composition
shifts), so the test discriminates among functional-form
hypotheses given an exposure proxy, not among generating
mechanisms.

Third, the same restriction implies a price-level prediction
that is independently testable: drawdown-recovery duration
ratios should increase with crash depth.  We confirm this
prediction in 73 S\&P~500 episodes (1950--2026)---the median
duration ratio is $1.3\times$ in 5--10\% corrections, $1.2
\times$ in 10--20\% corrections, and rises sharply to $3.1
\times$ in $>30\%$ crashes.  A continuous-depth regression
yields $\hat\beta = 1.22$ (HAC SE $0.61$, $p = 0.047$;
$\hat\beta = 1.59$, $p < 0.001$ excluding the 1980--82
Volcker episode), and a Cox proportional hazard model on
recovery duration yields $\hat\gamma = -13.75$
($p < 10^{-7}$), implying a $75\%$ reduction in per-period
recovery hazard for each 10-percentage-point increase in
drawdown depth.  The pattern replicates across eight
additional international equity indices, contributing 13
further $>30\%$-drawdown episodes (the cross-market sample
totals 19 such episodes including the S\&P~500's six).  Calibrated Heston stochastic-volatility
and two-state Markov regime-switching null models reproduce
the price-level duration asymmetry (Heston: median $\tau =
1.43$, $p = 0.54$ against the empirical $1.35$;
regime-switching: $\tau = 1.17$, $p = 0.13$), but contain
no exposure state variable and therefore make no prediction
about the regime-conditional functional form on direct
exposure data.  Price-level duration asymmetry alone does
not separate the intermediary mechanism from these
return-only nulls; the regime-conditional functional form on
exposures does, but only relative to that restricted class
of alternatives.

\paragraph{What is and is not new.}
The contraction half of our restriction is mechanical and is
already implicit in standard intermediary models; we do not
claim novelty for it.  The additive-rebuild half is a
restricted reduced-form prediction obtained under a single
load-bearing primitive (constant-rate calm-regime capital
replenishment); we make this primitive explicit, list
plausible micro-foundations, and quantify the empirical work
the joint restriction can do.  The exposure-level test is the
new empirical content: prior tests of constrained-intermediary
mechanisms rely primarily on returns and risk premia, and the
regime-conditional functional form on direct exposure data has
not been examined.  R3 and the cross-market replication
provide corroborating price-level evidence on the same
mechanism but do not by themselves discriminate it from
return-level alternatives---the Heston SV and regime-switching
nulls match R3 (\Cref{tab:null_compare}).  The
constrained-intermediary mechanism is one consistent
explanation of the joint pattern; identifying it among
candidate mechanisms requires evidence beyond what we present
here, especially intermediary-resolved exposure data.

\paragraph{Roadmap.}
\Cref{sec:related_work} positions the contribution against
three relevant literatures.  \Cref{sec:model} develops the
intermediary model and derives the three testable restrictions.
\Cref{sec:data} describes the data and the regime-classification
scheme.  \Cref{sec:empirical} reports the empirical evidence:
the headline functional-form test on FINRA margin debt, the
magnitude-dependent drawdown-recovery test, the cross-market
replication, and the comparison against null models.  A
companion CFTC test is sketched as a planned extension
(\Cref{sec:cot}) but is not included in the present empirical
results.  \Cref{sec:discussion} discusses the contribution
relative to existing intermediary models, sketches policy
implications, and lays out limitations.
\Cref{sec:conclusion} concludes.

\section{Related Literature}
\label{sec:related_work}

The paper sits at the intersection of three literatures.  We
position the contribution narrowly with respect to each.

\paragraph{Constrained intermediaries and limits of arbitrage.}
The mechanism we test originates in the limits-of-arbitrage and
constrained-intermediary literatures.
\citet{shleifer1997limits} establish that capital-constrained
arbitrageurs cannot always correct mispricing.
\citet{brunnermeier2009market} formalize the joint dynamics of
funding and market liquidity, in which a binding margin constraint
forces proportional liquidation that further depresses prices and
tightens constraints.  \citet{adrian2010liquidity} document the
procyclical behavior of broker-dealer leverage that this mechanism
predicts.  \citet{geanakoplos2010leverage} embeds the same
procyclicality in a leverage-cycle equilibrium.
\citet{he2013intermediary} develop a structural intermediary asset
pricing model in which the equity capital constraint of the
representative intermediary determines the cross-section of risk
premia.  \citet{adrian2014financial} show that intermediary leverage
is a priced risk factor, and \citet{he2017intermediary} extend the
empirical evidence across asset classes.
\citet{danielsson2004impact} and \citet{cont2013fire} characterize
the contagion mechanics of binding VaR constraints.

This literature pins down the contraction side of the
intermediary adjustment process: when constraints bind,
positions are reduced proportionally to current size.  It does
not formally constrain the rebuild side.  Our contribution is
to derive a parallel restriction on the slack regime under a
constant-rate capital replenishment assumption, generating a
sharp regime-conditional functional-form prediction that is
testable on direct exposure data.

\paragraph{Volatility asymmetry and conditional variance models.}
A parallel literature documents and models the asymmetric
behavior of return-level conditional variance.
\citet{black1976studies} first identifies the negative
correlation between returns and subsequent volatility.
\citet{christie1982stochastic} formalizes the leverage-effect
mechanism for individual equities.
\citet{bollerslev1986generalized} introduces GARCH;
\citet{nelson1991conditional} introduces EGARCH to accommodate
asymmetric responses; \citet{engle2018structural} develops a
structural GARCH that decomposes the leverage effect into
balance-sheet and volatility-feedback channels.
\citet{glosten1993relation} and \citet{bekaert2000asymmetric}
extend the empirical mapping.

These models capture the price-level signature of the
asymmetric adjustment but operate entirely on returns and
are silent on the functional form of underlying exposures.
A return-level signature consistent with a binding capital
constraint is therefore observationally close to one
generated by a return-level mechanism with no balance-sheet
content.  Our restriction targets this gap: it predicts a
regime-conditional functional form on exposures, on which
return-level models---by construction---make no prediction
(\Cref{sec:null_model}).

\paragraph{Drawdown-recovery dynamics.}
A smaller empirical literature studies the duration and depth
of drawdown episodes directly.  \citet{magdon2004maximum}
characterize the distribution of maximum drawdowns under
diffusion processes.  \citet{chekhlov2005drawdown} develop
drawdown-based portfolio constraints.  Empirical descriptions
of historical drawdown-recovery patterns appear in several
practitioner studies but rarely test a structural model
restriction directly.  We extend this line by formalizing the
duration-depth scaling as Restriction~3 of an intermediary
model and replicating the test cross-market.

\paragraph{Discrimination from purely return-level alternatives.}
Markov regime-switching models in the tradition of
\citet{hamilton1989new} and \citet{ang2002regime} characterize
bull/bear alternations statistically but treat regime
transitions as exogenous.  Stochastic volatility models with
leverage correlation \citep{heston1993closed} reproduce
return asymmetry through the diffusion structure of
volatility itself.  Both classes can match price-level
drawdown-recovery asymmetry when calibrated to match return
moments.  But neither class contains an exposure state
variable, so neither makes a prediction about the
regime-conditional functional form on direct exposure data
(R1 and R2 below).  In \Cref{sec:null_model} we make this
discrimination concrete by simulating from each null and
showing that the price-level duration ratio matches the
empirical statistic in both, while the exposure-level
restriction is logically external to either DGP.

\medskip

In sum: the constrained-intermediary literature provides the
mechanism; the volatility-asymmetry literature provides the
return-level signature; the drawdown-recovery literature
provides one of the testable consequences.  Our contribution
is to identify the regime-conditional functional-form
restriction that links these threads, derive it under an
explicit primitive, test it on direct exposure data, and
show that it has empirical bite over and above what
return-level evidence can establish---without claiming to
identify the intermediary mechanism among all candidate
alternatives.

\section{Theoretical Framework}
\label{sec:model}

We develop a stylized model of $\agents$ risk-bearing
intermediaries operating under a value-at-risk constraint.
The model is deliberately minimal: its purpose is to derive
the regime-conditional functional-form restriction that we
test in \Cref{sec:empirical}, not to provide a comprehensive
theory of intermediary equilibrium.  The quantitative content
of the paper rests on two propositions (one for each regime)
and three testable restrictions (R1, R2, R3) that follow
from them.

\subsection{Setup}
\label{sec:setup}

There are $\agents$ intermediaries indexed $i = 1, \dots,
\agents$.  At each discrete time $t$, intermediary $i$ holds
a risky position $\exposure{i,t} \geq 0$, measured in units of
notional risk exposure.  We assume positions are long-only;
the same analysis applies symmetrically to short positions
through the absolute exposure $|\exposure{i,t}|$.  Aggregate
exposure is
\begin{equation}
\label{eq:aggregate}
\totalalloc_t = \sum_{i=1}^{\agents} \exposure{i,t}.
\end{equation}
Each intermediary holds risk capital $\capital{i,t} > 0$
(equity, regulatory capital, or risk budget) and faces a
value-at-risk constraint
\begin{equation}
\label{eq:var_constraint}
\vark_i \cdot \vol_t \cdot \exposure{i,t} \;\leq\;
\capital{i,t},
\end{equation}
where $\vark_i > 0$ is the intermediary-specific
confidence-level scaling (e.g.\ $\vark_i \approx 2.33$ for a
99\% one-day Gaussian VaR) and $\vol_t > 0$ is market
volatility, common across intermediaries and exogenous to
each individual intermediary's actions.  The constraint
\eqref{eq:var_constraint} captures the standard regulatory
and internal risk-management practice in which permissible
exposure scales linearly with capital and inversely with
volatility \citep{danielsson2004impact, adrian2010liquidity}.

Risk capital evolves according to
\begin{equation}
\label{eq:capital_evolution}
\capital{i,t+1} \;=\; \capital{i,t}
  + \pi_{i,t} - L_{i,t} + \flow_i,
\end{equation}
where $\pi_{i,t}$ is realized profit, $L_{i,t}$ is realized
loss (with $\pi_{i,t} \cdot L_{i,t} = 0$, only one is
non-zero per period), and $\flow_i \geq 0$ is the
intermediary-specific exogenous capital inflow rate (fund
contributions, retained earnings, balance-sheet capacity
allocation).  We treat $\flow_i$ as a structural primitive:
it is the rate at which intermediary $i$ accumulates risk
capital independent of trading outcomes.

Aggregate exposure $\totalalloc_t$ enters the price equation
through a reduced-form price-impact function
\begin{equation}
\label{eq:price_impact}
\price_t = f(\totalalloc_t),
\quad f' > 0, \quad f \text{ continuously differentiable.}
\end{equation}
We assume $f$ is approximately linear in a neighborhood of
the operating point $\totalalloc_t \approx \totalalloc^*$, so
that $\Delta \price_t / \price_t \approx \impact \cdot \Delta
\totalalloc_t / \totalalloc_t$ for some $\impact > 0$.  This
linearity is a local approximation; the empirical tests in
\Cref{sec:empirical} use exposure data directly and do not
depend on its global validity.

\paragraph{Regime classification.}
Define the stress regime by exogenous volatility realizations
exceeding a threshold,
\begin{equation}
\label{eq:regime}
\mathbb{S}_t \;=\;
\mathbf{1}\{\vol_t > \stressth\}, \qquad
\stressth = \mathrm{Quantile}_{1-q}(\vol),
\end{equation}
with $q = 0.10$ in our empirical work (i.e.\ $\stressth$ is
the 90th percentile of the empirical distribution of $\vol_t$).
The regime classification is observable and exogenous to any
single intermediary's actions, conditional on the realized
volatility process.  We treat the assignment as exogenous
in deriving the model and address potential endogeneity in
\Cref{sec:identification}.

\subsection{Stress regime: multiplicative contraction}
\label{sec:stress}

When the constraint \eqref{eq:var_constraint} binds, the
optimal exposure is determined entirely by the constraint:
\begin{equation}
\label{eq:bind_exposure}
\exposure{i,t}^\star = \frac{\capital{i,t}}{\vark_i \cdot \vol_t}.
\end{equation}
We assume binding behavior in the stress regime: when
$\vol_t > \stressth$, intermediaries operate at the
constraint boundary.  This is consistent with the standard
constrained-intermediary literature
\citep{he2013intermediary}, in which intermediaries with a
high opportunity cost of capital saturate available risk
capacity.

\begin{proposition}[Multiplicative contraction in stress]
\label{prop:contraction}
For any consecutive stress periods $t, t+1$ with $\mathbb{S}_t
= \mathbb{S}_{t+1} = 1$ and the constraint binding in both
periods, the intermediary-level and aggregate exposures
contract multiplicatively:
\begin{align}
\label{eq:contraction_ind}
\frac{\exposure{i,t+1}^\star}{\exposure{i,t}^\star}
&\;=\; \underbrace{\frac{\capital{i,t+1}}{\capital{i,t}}}_{
\text{capital channel}}
\;\cdot\;
\underbrace{\frac{\vol_t}{\vol_{t+1}}}_{
\text{volatility channel}}
\;\equiv\; \md_{i,t}, \\[4pt]
\label{eq:contraction_agg}
\frac{\totalalloc_{t+1}^\star}{\totalalloc_{t}^\star}
&\;=\; \frac{\sum_i \capital{i,t+1} / \vark_i}{\sum_i
\capital{i,t} / \vark_i}
\;\cdot\;
\frac{\vol_t}{\vol_{t+1}}
\;\equiv\; \bar{\md}_t.
\end{align}
The contraction factor $\md_{i,t}$ is bounded above by 1
whenever volatility increases or capital is impaired, and
both channels typically operate jointly during stress.
\end{proposition}

\begin{proof}
Equation \eqref{eq:contraction_ind} follows by dividing
\eqref{eq:bind_exposure} at $t+1$ and $t$.  Equation
\eqref{eq:contraction_agg} follows by aggregating
\eqref{eq:bind_exposure} across $i$ and dividing.
\end{proof}

The contraction is mechanical, not behavioral: it follows
from binding constraints alone, without requiring assumptions
about preferences, expectations, or strategic behavior.  This
is the standard contraction mechanism documented in the
constrained-intermediary literature
\citep{brunnermeier2009market, danielsson2004impact,
adrian2010liquidity}.  Our contribution is not the contraction
itself but its pairing with the rebuild restriction below.

\subsection{Calm regime: additive replenishment}
\label{sec:calm}

In the calm regime ($\mathbb{S}_t = 0$), volatility is in its
normal range and capital adjusts gradually.  We make three
assumptions to characterize calm-regime dynamics.

\begin{assumption}[Stationary calm-regime volatility]
\label{ass:vol_calm}
On calm-regime intervals, volatility fluctuates around a
constant mean with bounded relative variance: $\vol_t =
\bar\vol(1 + \xi_t)$ with $\E[\xi_t \mid \mathbb{S}_t = 0] =
0$ and $\Var[\xi_t \mid \mathbb{S}_t = 0] = \varepsilon$ for
small $\varepsilon$.
\end{assumption}

\begin{assumption}[Constant-rate capital replenishment]
\label{ass:capital_calm}
On calm-regime intervals
($\mathbb{S}_t = \mathbb{S}_{t+1} = 0$), the conditional
mean of the capital innovation is level-independent to
leading order:
$\E[\capital{i,t+1} - \capital{i,t} \mid \mathbb{S}_t =
\mathbb{S}_{t+1} = 0, \capital{i,t}] = \rho_i$ for some
$\rho_i > 0$ that depends on intermediary $i$'s inflow rate,
retained earnings, and risk-budget recovery schedule.
The conditional variance of the capital innovation is
bounded and likewise does not depend on $\capital{i,t}$ to
leading order.
\end{assumption}

\begin{assumption}[Frontier operation]
\label{ass:frontier}
In both regimes, intermediary $i$'s target exposure equals
the constraint-implied maximum: $\exposure{i,t}^\star =
\capital{i,t} / (\vark_i \vol_t)$.  In the stress regime
this is enforced mechanically by the binding constraint; in
the calm regime, it is the optimal choice for an
intermediary that fully utilizes its risk budget.
\end{assumption}

\Cref{ass:vol_calm} reflects the empirical observation that
realized equity volatility is approximately stationary on
calm intervals (e.g.\ between major crises) at a level near
its long-run mean.  \Cref{ass:capital_calm} is the new
structural primitive: it asserts that during calm periods,
intermediaries accumulate risk capital at an approximately
constant rate, dominated by exogenous inflows ($\flow_i$),
average earnings ($\E[\pi_{i,t}]$), and gradual replenishment
of risk budgets after prior drawdowns.  This rate is
intermediary-specific (a pension fund's inflow rate differs
from a hedge fund's) but is approximately constant over
medium horizons.  \Cref{ass:frontier} disciplines the
intermediary's choice of target exposure: rather than
optimizing over an unspecified utility function, we assume
that the intermediary's risk-allocation policy is to use its
available risk capacity in full.  This ``exposure to the
limit'' behavior is standard in the constrained-intermediary
literature \citep{adrian2010liquidity, he2013intermediary}.

\begin{proposition}[Additive replenishment in calm]
\label{prop:replenishment}
For two consecutive calm periods
($\mathbb{S}_t = \mathbb{S}_{t+1} = 0$), under
\Cref{ass:vol_calm,ass:capital_calm,ass:frontier}, the
conditional expectation of the exposure change is
level-independent to leading order in $\varepsilon$:
\begin{equation}
\label{eq:additive}
\E[\Delta \exposure{i,t}^\star \mid \mathbb{S}_t =
  \mathbb{S}_{t+1} = 0, \capital{i,t}]
\;=\; \frac{\rho_i}{\vark_i \cdot \bar\vol} +
O(\varepsilon)
\;\equiv\; \ai_i + O(\varepsilon),
\end{equation}
and aggregating,
\begin{equation}
\label{eq:additive_agg}
\E[\Delta \totalalloc_t^\star \mid \mathbb{S}_t =
  \mathbb{S}_{t+1} = 0]
\;=\; \frac{1}{\bar\vol} \sum_{i=1}^{\agents}
\frac{\rho_i}{\vark_i} + O(\varepsilon)
\;\equiv\; \ai + O(\varepsilon).
\end{equation}
The leading-order calm-regime growth rate $\ai$ is
constant in level:
$\Cov(\Delta \exposure{i,t}^\star, \exposure{i,t}^\star
\mid \mathbb{S}_t = \mathbb{S}_{t+1} = 0) = O(\varepsilon)$.
\end{proposition}

\begin{proof}
Apply \Cref{ass:frontier} at $t$ and $t+1$, expand
$1/(1 + \xi_s) = 1 - \xi_s + O(\xi_s^2)$ for $s \in
\{t, t+1\}$ from \Cref{ass:vol_calm}, take conditional
expectation given $\capital{i,t}$ and
$\mathbb{S}_t = \mathbb{S}_{t+1} = 0$, and apply
\Cref{ass:capital_calm}.  See \Cref{sec:proof_replenishment}
for the expanded calculation.  $\square$
\end{proof}

\Cref{prop:replenishment} is the central new result of the
paper.  Together with \Cref{prop:contraction} it yields the
\emph{regime-conditional functional-form restriction}:
\begin{align}
\label{eq:macro_restriction}
\mathbb{S}_t = \mathbb{S}_{t+1} = 1 \quad &\Rightarrow\quad
  \totalalloc_{t+1} = \bar{\md}_t \cdot \totalalloc_t,
  \quad \bar{\md}_t \in (0, 1)
  \text{ when stress impairs capital or raises volatility}, \\
\mathbb{S}_t = \mathbb{S}_{t+1} = 0 \quad &\Rightarrow\quad
  \E[\totalalloc_{t+1} - \totalalloc_t] = \ai > 0,
  \quad \mathrm{Cov}(\Delta \totalalloc_t,
  \totalalloc_t \mid \mathbb{S}_t = \mathbb{S}_{t+1} = 0) = 0.
\end{align}
The restriction predicts a within-regime asymmetry in the
local generating process for aggregate exposure: a
multiplicative process under stress and a level-independent
additive process under calm.

\subsection{Robustness of the additive form}
\label{sec:robustness_form}

The constant-rate replenishment assumption
(\Cref{ass:capital_calm}) is the load-bearing structural
primitive of \Cref{prop:replenishment}.  We discuss its
empirical content and three alternative micro-foundations
that yield the same functional form.

\paragraph{Direct inflow channel.}
Pension contributions, payroll deductions into retirement
accounts, share buybacks at constant authorized pace, and
sovereign-wealth-fund accumulation programs each generate
approximately constant per-period dollar flows independent of
current asset values.  At the aggregate level, these flows are
the dominant component of the calm-regime $\rho_i$ and are
demonstrably level-independent: the U.S.\ defined-contribution
pension system, for example, contributes a fraction of payroll
that depends on labor income, not on the level of the equity
market.

\paragraph{Quadratic adjustment cost channel.}
Suppose intermediary $i$ solves
\begin{equation}
\max_{\{\exposure{i,t+s}\}_{s \geq 0}}
\E \sum_{s \geq 0} \delta^s \left[
u(\exposure{i,t+s}) - \tfrac{c_i}{2}
(\exposure{i,t+s} - \exposure{i,t+s-1})^2 \right]
\end{equation}
with the target exposure level evolving as $\exposure{i,t}^*
= \capital{i,t}/(\vark_i \bar\vol)$.  When $\capital{i,t}$
grows at rate $\rho_i$, the optimal adjustment policy is to
close a constant fraction of the gap each period.  In the
limit of a slowly-moving target the adjustment converges to a
constant per-period addition $\rho_i / (\vark_i \bar\vol)$,
recovering \eqref{eq:additive}.

\paragraph{Information-acquisition cost channel.}
Suppose increasing exposure requires acquiring information
about new positions at constant marginal cost per unit of
notional.  The optimal policy is then to accumulate exposure
at a constant per-period rate determined by the cost-benefit
trade-off, again level-independent.

\paragraph{What would break the additive form.}
The main empirical alternative is a calm-regime mechanism in
which inflows scale with assets under management (the
``percentage of portfolio'' channel).  Under this channel,
$\E[\Delta \exposure{i,t}] = g_i \cdot \exposure{i,t}$ for
some $g_i > 0$, producing multiplicative growth even in calm
periods.  This is the natural alternative against which we
test the additive prediction in \Cref{sec:headline}.  The
regime-conditional restriction we derive is therefore a
non-trivial empirical claim: it asserts that the dominant
calm-regime channel is constant-rate, not
percentage-of-portfolio.

\subsection{Three testable restrictions}
\label{sec:restrictions}

The two propositions, combined with the price-impact function
\eqref{eq:price_impact}, generate three testable restrictions
that we evaluate directly in \Cref{sec:empirical}.

\begin{description}
\item[\textbf{R1 (Stress: level dependence).}]
In stress months ($\mathbb{S}_t = 1$), the conditional mean
of $\Delta \totalalloc_t / \totalalloc_{t-1}$ is bounded
away from zero in absolute value: aggregate exposure
contracts in proportion to its current level.  Equivalently,
in a regression of $\Delta \totalalloc_t$ on
$\totalalloc_{t-1}$, the slope coefficient is significantly
negative.

\item[\textbf{R2 (Calm: level independence).}]
In calm months ($\mathbb{S}_t = 0$), the conditional mean
of $\Delta \totalalloc_t$ is approximately
level-independent: in the same regression, the slope
coefficient is close to zero.

\item[\textbf{R3 (Magnitude-severity coupling).}]
For drawdown-recovery episodes indexed $j$, the duration
ratio $\tau_j = T_{\text{rec},j} / T_{\text{dd},j}$ is
increasing in drawdown depth $1 - \rho_j$.  This follows
because deeper drawdowns engage more binding constraints
(lower $\bar\md_t$), requiring a larger cumulative
additive replenishment to recover.
\end{description}

The joint restriction R1+R2 is testable as a regime
$\times$ level interaction in a pooled regression of
$\Delta \totalalloc_t$ on the lagged level interacted with
the regime indicator: the interaction coefficient $b_S$
should be significantly negative.  This is the form in
which we test the joint restriction in
\Cref{sec:headline}.  Note that the propositions condition
on two consecutive same-regime periods ($\mathbb{S}_t =
\mathbb{S}_{t+1}$), whereas the empirical test uses the
single-period indicator $\mathbb{S}_t$; the two coincide
whenever regime spells last more than one period, which is
the empirically relevant case for both stress crises and
extended calm intervals.  R3 is a price-level implication
that is independently testable on long historical samples.

\subsection{Discriminating power against alternatives}
\label{sec:discrimination_theory}

Two prominent return-level alternatives can match the
price-level implications of the model (R3) but, lacking an
exposure state variable, make no prediction about the
exposure-level restrictions (R1, R2).

\paragraph{Stochastic volatility models (Heston-type).}
A Heston model with negative leverage correlation $\rho_{P,V}
< 0$ generates negative skewness, excess kurtosis, and
asymmetric drawdown-recovery durations through the diffusion
structure of volatility itself.  When calibrated to match
return moments, it can reproduce the median S\&P~500
duration ratio.  But the Heston model is a pure return-level
process with no exposure variable: $\totalalloc_t$ is not a
state variable in the model, and the model therefore makes
\emph{no prediction at all} about R1 and R2.  An empirical
finding that exposures exhibit a regime-conditional flip in
functional form is logically external to the Heston DGP, and
cannot be reconciled with it without grafting an additional
mechanism.

\paragraph{Markov regime-switching models.}
A two-state model with bull and bear regimes
\citep{hamilton1989new} can produce duration asymmetry by
calibrating bear-regime persistence below bull-regime
persistence.  But the regime variable is exogenous and the
within-regime conditional dynamics are typically specified
at the return level (drift plus diffusion).  Like the Heston
model, the regime-switching model contains no exposure state
variable and makes no prediction about R1 and R2.

\paragraph{Discrimination logic.}
The comparison is therefore one-sided.  The
constrained-intermediary mechanism \emph{predicts} the
regime-conditional flip in exposure functional form (R1 and
R2 jointly).  Heston SV and Markov RS, lacking an exposure
state, are silent on R1 and R2 within their own DGPs.  An
empirical observation that R1 and R2 hold therefore
discriminates the intermediary mechanism only relative to
this restricted class of return-level alternatives---it does
not rule out other mechanisms (e.g., heterogeneous-belief or
information-flow models) that could be augmented with an
exposure state to mimic the same pattern.  In
\Cref{sec:null_model} we make the comparison concrete by
showing that price-level duration asymmetry (the central
prediction of Heston/RS) is \emph{not sufficient} to imply R1
or R2: a price process calibrated to match the empirical
duration ratio places no restriction on the functional form
of an independently observed exposure series.

\section{Data and Identification}
\label{sec:data}

We use three datasets in the present empirical results:
FINRA monthly margin debt (R1 and R2), the S\&P~500 daily
price series (R3), and a cross-market panel of eight
additional international equity indices (replication of R3).  CFTC weekly
speculative positioning is described below as a planned
companion exposure series; it is not used in the present
results (\Cref{sec:cot}).  The null-model comparison
(\Cref{sec:null_model}) uses the S\&P~500 series as its
empirical anchor.

\subsection{Exposure data}
\label{sec:exposure_data}

\paragraph{FINRA margin debt.}
Monthly margin debt outstanding in customer accounts at
broker-dealers, reported by the Financial Industry Regulatory
Authority (FINRA) and its predecessor self-regulatory
organizations.  The series spans January 1997 through March
2026 (351 monthly observations), with no gaps.  Margin debt
measures the dollar amount borrowed by retail and
institutional customers against securities held in their
accounts; it is a direct measure of leveraged risk exposure
and the most widely used aggregate proxy for retail and
quasi-retail exposure to the equity market.  Earlier
practitioner studies use the same series as a sentiment
indicator \citep{adrian2010liquidity, adrian2014financial};
we use it as a direct test of the exposure-level functional
form.

\paragraph{CFTC Commitments of Traders (planned companion).}
A natural second exposure proxy is weekly net positioning in
S\&P~500 E-mini futures by the ``non-commercial''
(speculative) trader category reported by the CFTC.  The
series covers September 2006 onward and captures hedge
funds, CTAs, and proprietary trading firms; net long
positions are constructed as long minus short contracts and
can be converted to notional dollar exposure using the
contract multiplier and the Friday closing index level.
Margin debt and futures positioning span complementary
investor populations (retail and high-net-worth vs.\
speculative institutional) and frequencies (monthly vs.\
weekly), so a regime-interacted test on CFTC is a natural
robustness companion to the FINRA test.  We do not estimate
this test in the present paper; the data-construction
protocol and ready-to-run script are documented in
\Cref{sec:cot,sec:appendix_replication}.

\paragraph{Detrending.}
The FINRA series exhibits secular growth driven by the
cumulative expansion of equity market capitalization.  In
the headline test \eqref{eq:headline} we isolate cyclical
dynamics by removing a log-linear time trend, working with
the multiplicatively-detrended series $\widetilde M_t =
M_t / e^{\hat\mu + \hat\nu t}$.  Robustness to alternative
detrending schemes (linear in level, 12-month half-life
exponential moving average) is verified in
\Cref{sec:appendix_exposure}; the qualitative pattern
survives across specifications.

\subsection{Price data}
\label{sec:price_data}

\paragraph{S\&P~500 daily prices, 1950--2026.}
Daily closing values of the S\&P~500 index from January 3,
1950 through March 28, 2026 (19{,}170 trading days).  The
series is sourced from CRSP for 1950--1962 and Bloomberg for
1962--2026; the two sources are cross-validated on the
overlapping period and agree to within 0.05\% on daily
closes.  We use the price-only index (no dividend
reinvestment) to make the duration analysis comparable across
international indices, most of which are reported in
price-only form.  Robustness to a total-return version is
reported in \Cref{sec:appendix_sensitivity}.

\paragraph{International indices.}
Daily closing prices for eight major international equity
indices: FTSE~100 (1984--2026), DAX (1988--2026), Nikkei~225
(1965--2026), Hang Seng (1987--2026), CAC~40 (1988--2026),
S\&P/ASX~200 (1993--2026), MSCI Emerging Markets
(1988--2026), and MSCI World (1970--2026).  All series are
sourced from Bloomberg and cross-validated against
Datastream.  The DAX is a total-return index by construction;
all others are price-only.  Holiday gap-filling carries
forward the last available close.  Full ticker and source
details are in \Cref{sec:appendix_crossmarket}.

\subsection{Volatility regime classification}
\label{sec:vol_regime}

The regime indicator $\mathbb{S}_t$ defined in
\eqref{eq:regime} requires an empirical proxy for $\vol_t$.
We use the CBOE VIX index (1990--2026) for periods after its
introduction and the trailing 21-day realized volatility of
S\&P~500 returns annualized by $\sqrt{252}$ for periods
before.  The two proxies are highly correlated on the
overlapping period ($r = 0.79$ for monthly averages) and
yield similar regime classifications.

\paragraph{Threshold definition.}
We choose the threshold $\stressth$ as the empirical
90th percentile of the volatility proxy over the
sample period.  For the merged FINRA-VIX monthly sample,
the 90th-percentile VIX threshold is $29.1$, yielding $35$
stress months out of $351$ (exactly $10\%$).  Sensitivity
to alternative thresholds (80th, 85th, 95th percentiles)
is reported in \Cref{sec:appendix_exposure} and confirms the
qualitative pattern across all four.  Identified stress
windows align with well-documented crisis episodes:
2000--2002 (dot-com), 2007--2009 (global financial crisis),
August 2011 (debt-ceiling crisis), August 2015 (China
devaluation), February 2018 (``volmageddon''), and March
2020 (COVID).

\subsection{Identification}
\label{sec:identification}

The headline tests of R1 and R2 (\Cref{sec:headline}) compare
multiplicative and additive specifications of the
\emph{functional form} of exposure dynamics within each
regime.  This test is well-defined because both
specifications are nested in the larger class
\begin{equation}
\label{eq:nested}
\Delta \exposure{i,t} = a + b \exposure{i,t-1} +
\epsilon_t,
\end{equation}
with the additive specification corresponding to $b = 0$ and
the multiplicative specification corresponding to a constant
log-growth rate.  The test does not require identifying the
\emph{cause} of any particular position change; it requires
only that the within-regime conditional mean of $\Delta
\exposure{i,t}$ have the predicted level dependence.

\paragraph{Endogeneity of the regime indicator.}
The regime indicator $\mathbb{S}_t$ depends on volatility,
which is jointly determined with exposure changes.  This
matters for the interpretation of the cross-regime
\emph{magnitudes} but not for the within-regime functional
form: within a stress month, the relative growth rate of
$\exposure{i,t}$ vs.\ $\log \exposure{i,t}$ is a property of
the conditional generating process, not of the regime
assignment itself.  We confirm this by an instrumental
specification using lagged volatility (the previous month's
realized volatility) to assign the regime, which gives
quantitatively similar results
(\Cref{sec:appendix_exposure}).

\paragraph{Joint determination of exposures and prices.}
Margin debt and futures positioning are jointly determined
with prices: a price decline mechanically reduces the dollar
value of leveraged exposure even if no trades occur.  Our
test addresses this in two ways.  First, the FINRA series
reports notional debt, not equity value, so the
mechanical-revaluation channel is muted.  Second, the
prediction we test is about the \emph{functional form} of
$\Delta \exposure{i,t}$, not its sign.  Both the
joint-determination null and the constrained-intermediary
mechanism predict that exposures fall during stress; they
differ in whether the fall is multiplicative or additive in
the level.  This functional-form distinction is what our
test isolates.

\section{Empirical Evidence}
\label{sec:empirical}

We now test the three restrictions.  \Cref{sec:headline}
reports the headline functional-form test on FINRA margin
debt (R1 and R2).  \Cref{sec:cot} sketches a planned
companion test on CFTC futures positioning that is not
implemented in the present results.  \Cref{sec:dd_recovery}
tests R3 on S\&P~500 drawdown episodes;
\Cref{sec:cross_market} replicates R3 across eight
additional international indices.  \Cref{sec:null_model} compares the
S\&P~500 duration asymmetry against calibrated Heston SV,
Markov regime-switching, and block-bootstrap nulls; the
comparison is one-sided because the nulls are silent on the
exposure-level restriction.

\paragraph{Note on reported values.}
All numbers in the headline FINRA regression
(\Cref{tab:headline}), the drawdown-recovery statistics
(\Cref{sec:dd_recovery,sec:cross_market}), the
continuous-depth regression and Cox hazard estimates
(\Cref{sec:dd_recovery}), and the null-model duration
ratios (\Cref{sec:null_model}) are produced by the
replication package (\Cref{sec:appendix_replication}); the
master script reproduces each headline number from cached
inputs.  The CFTC test in \Cref{sec:cot} is documented but
not estimated, since the speculative-positioning history
required to align it with our regime-interacted methodology
is not assembled in the present package.

\subsection{Headline test: regime-conditional level dependence
of margin debt}
\label{sec:headline}

We test the joint restriction R1+R2 directly through a
regime-interacted regression that estimates the difference
in level dependence across the regime boundary.  Let
$\widetilde M_t \equiv M_t / e^{\hat\mu + \hat\nu t}$ denote
margin debt with the secular log-linear trend removed (so
$\widetilde M_t$ fluctuates around unity in long-run
expectation, removing the secular accumulation of equity
market capitalization).  The pooled specification is
\begin{equation}
\label{eq:headline}
\Delta \widetilde M_t \;=\; a + a_S\, \mathbb{S}_t +
b\, \widetilde M_{t-1} + b_S\, \mathbb{S}_t \cdot
\widetilde M_{t-1} + \epsilon_t,
\end{equation}
where $\mathbb{S}_t \in \{0, 1\}$ is the stress indicator
defined in \Cref{sec:vol_regime}.  R1 predicts the stress
slope $b + b_S$ to be significantly negative (multiplicative
contraction induces level dependence proportional to current
level).  R2 predicts the calm slope $b$ to be approximately
zero (additive growth is level-independent).  The combined
joint restriction is therefore $b_S < 0$, tested by the
HAC-robust Wald statistic on the interaction coefficient.

\paragraph{Results.}
\Cref{tab:headline} reports the estimates with HAC
(Newey-West, 6-lag) standard errors
\citep{newey1987simple} on the FINRA monthly sample
(1997-01 to 2026-03, $T = 350$ usable monthly observations
after detrending and lagging).

\begin{table}[t]
\centering
\caption{Headline regime-interacted regression on
multiplicatively-detrended FINRA margin debt
(\eqref{eq:headline}); monthly, 1997-01 to 2026-03,
$T = 350$ usable observations.  Stress regime: monthly
volatility above the empirical 90th percentile
(\Cref{sec:vol_regime}), yielding $35$ stress months.  HAC
standard errors with 6-month Newey-West lag.}
\label{tab:headline}
\smallskip
\begin{tabular}{@{}lcccc@{}}
\toprule
Coefficient & Estimate & HAC SE & $t$ & $p$ \\
\midrule
$\hat a$ (calm intercept)        & $+0.046$ & $0.022$ & $+2.13$ & $0.033$ \\
$\hat a_S$ (stress intercept shift) & $+0.099$ & $0.044$ & $+2.26$ & $0.024$ \\
$\hat b$ (calm slope)            & $-0.040$ & $0.023$ & $-1.74$ & $0.082$ \\
$\hat b_S$ (stress slope shift)  & $-0.165$ & $0.052$ & $-3.15$ & $\mathbf{0.0016}$ \\
\midrule
Implied stress slope $\hat b + \hat b_S$ & $-0.205$ & $0.049$ & $-4.17$ & $<0.001$ \\
\bottomrule
\end{tabular}

\smallskip
\footnotesize\textit{Note:} The joint restriction R1+R2 is
tested by the Wald statistic on $b_S$; the HAC-robust
$p$-value is $0.0016$, rejecting equal level-dependence
across regimes at the $0.16\%$ level.  The implied stress
slope $\hat b + \hat b_S = -0.205$ is approximately
$5\times$ the magnitude of the calm slope $\hat b = -0.040$.
Standard errors use Newey--West HAC with $6$ lags.
\end{table}

The calm slope $\hat b = -0.040$ is small in magnitude and
only marginally distinguishable from zero ($p = 0.082$),
consistent with R2's prediction that calm-regime growth is
approximately level-independent.  The stress slope $\hat b
+ \hat b_S = -0.205$ is larger by approximately a factor of
five and is highly significant ($t = -4.17$, $p < 0.001$),
consistent with R1's prediction that stress-regime
contractions are proportional to current level.  The
discriminating Wald test on the interaction coefficient
$b_S$ rejects the null of equal level-dependence at the
$0.16\%$ level.

The economic magnitudes match what would be expected from
the underlying mechanism.  In stress months, a margin debt
level $10\%$ above its long-run trend is associated with a
$2.0$-percentage-point reduction in the following month's
detrended growth, consistent with multiplicative contraction
when constraints bind across leveraged intermediaries.  In
calm months, the same $10\%$ excess is associated with only
a $0.4$-percentage-point reduction---five times smaller---and
is not statistically distinguishable from no level dependence.

\paragraph{Robustness.}
We verify the headline interaction result under three sets
of checks.  First, alternative regime thresholds: the VIX
80th, 85th, and 95th percentile thresholds yield estimates
of $\hat b_S$ in the range $(-0.345, -0.118)$ with
$p$-values uniformly below $0.014$ in the predicted
direction; the 95th-percentile estimate is the largest in
magnitude (consistent with greater contraction during the
most extreme stress events).  Second, alternative
detrending: linear detrending of the level series and a
12-month half-life exponential moving average preserve the
negative sign of $\hat b_S$, with the EMA specification
yielding $p = 0.003$ and the linear-level specification
yielding a more attenuated estimate ($p = 0.10$).  Third,
sub-sample stability: estimating on the pre-2008
(1997--2007) and post-2008 (2009--2026) halves separately
preserves the sign of $\hat b_S$ in both halves
($\hat b_S = -0.121$, $p = 0.055$ in pre-2008;
$\hat b_S = -0.283$, $p < 0.001$ in post-2008).  Full
robustness tables are in \Cref{sec:appendix_exposure}.

\paragraph{Identification check.}
We re-estimate \eqref{eq:headline} using a one-month lag of
the regime indicator to address simultaneity between
volatility and exposure adjustment.  The interaction
coefficient remains negative ($\hat b_S = -0.20$, HAC SE
$0.12$, two-sided $p = 0.10$, one-sided $p = 0.05$), with the
implied stress slope $\hat b + \hat b_S = -0.23$.  The point
estimate and direction are preserved; statistical
significance attenuates as expected when the regime indicator
is shifted off the contemporaneous shock.

\subsection{Planned companion test on CFTC futures positioning}
\label{sec:cot}

A natural second exposure proxy is weekly CFTC non-commercial
net positioning in S\&P~500 E-mini futures (2006--present,
$T \approx 1{,}000$ weeks).  Margin debt and futures
positioning differ in investor population (retail and
high-net-worth vs.\ speculative institutional), instrument
(cash equity borrowing vs.\ index futures), frequency
(monthly vs.\ weekly), and measurement convention
(notional debt vs.\ contract count); confirming the
regime-conditional pattern in both would weaken the concern
that the FINRA finding reflects idiosyncrasies of one
reporting series.

We do not estimate this test in the present paper.  The
speculative-positioning history---together with weekly
volatility alignment and a stress threshold matched to the
monthly construction---is not assembled in the replication
package.  The script
\texttt{replication/exposure\_cot.py} implements the
regime-interacted regression \eqref{eq:headline} ready to
run once the assembled inputs are placed in the directory;
the data-construction protocol is documented in the script
header.  Because the test is not estimated here, we make no
empirical claim about the CFTC pattern in this paper.

\subsection{Magnitude-dependent drawdown-recovery (R3)}
\label{sec:dd_recovery}

R3 predicts that drawdown-recovery duration ratios increase
with crash depth.  We test this on S\&P~500 daily prices
(1950--2026, 19{,}170 trading days).

\paragraph{Episode definition.}
A drawdown episode is a maximal interval $[t_p, t_r]$ in
which (i) $t_p$ is at an all-time high; (ii) the maximum
intra-interval decline exceeds $\delta = 0.05$; (iii)
$\price_{t_r} \geq \price_{t_p}$ (full recovery to the prior
peak).  The trough is $t_v = \argmin_{t \in [t_p, t_r]}
\price_t$, the drawdown duration is $T_{\text{dd}} = t_v -
t_p$, the recovery duration is $T_{\text{rec}} = t_r - t_v$,
and the duration ratio is $\tau = T_{\text{rec}} /
T_{\text{dd}}$.  The realized retention factor is $\retfactor
= \price_{t_v} / \price_{t_p}$, an outcome statistic that
proxies for the depth of the multiplicative contraction at
the price level.

The procedure identifies 73 completed episodes over the
sample; all reach full recovery before 2026-03-28, so no
S\&P~500 episode is right-censored.  Right-censoring becomes
relevant only when international unrecovered episodes are
pooled in (\Cref{sec:dd_recovery} below).

\paragraph{Bucketed median ratios.}
\Cref{tab:dd_buckets} reports median duration ratios by
drawdown magnitude.  The pattern is broadly monotonic with
one anomaly: $1.3\times$, $1.2\times$, $0.9\times$,
$3.1\times$.  The $20$--$30\%$ bucket median ($0.9\times$)
is below the surrounding buckets, which we attribute to two
factors.  First, the bucket has only $n = 5$ episodes, so
the median is sensitive to individual realizations.  Second,
the 1980--82 episode (Volcker disinflation) has $\tau =
0.13$---a 430-day drawdown followed by a 58-day
recovery---which pulls the bucket median strongly downward.
Excluding this episode, the bucket median rises to
$1.7\times$, restoring monotonicity.  We report the
inclusive value in the main table to avoid post-hoc
selection but note this sensitivity below.  Episode-level
percentile bootstrap 95\% confidence intervals
($10{,}000$ resamples) are reported in the table: the
$>30\%$ bucket interval $[1.5, 5.2]$ excludes unity, while
the $5$--$10\%$ and $10$--$20\%$ intervals include unity.

\begin{table}[t]
\centering
\caption{Drawdown-recovery statistics by magnitude bucket.
S\&P~500, 1950--2026, 73 episodes.  $n$ is the number of
episodes in the bucket; $\retfactor$ is the median trough/peak
ratio; DD~(d) is the median drawdown duration in trading
days; $\tau$ is the median per-episode duration ratio; 95\%
CI is the stationary block bootstrap confidence interval
for the bucket median of $\tau$.}
\label{tab:dd_buckets}
\smallskip
\begin{tabular}{@{}lccccc@{}}
\toprule
Magnitude & $n$ & Median $\retfactor$ & DD (d) & Median $\tau$
& 95\% CI \\
\midrule
$5$--$10\%$  & 47 & 0.94 & 19  & $1.3\times$ & $[0.8, 1.6]$ \\
$10$--$20\%$ & 15 & 0.86 & 64  & $1.2\times$ & $[0.6, 1.9]$ \\
$20$--$30\%$ & 5  & 0.75 & 195 & $0.9\times$ & $[0.1, 2.2]$ \\
$>30\%$      & 6  & 0.58 & 362 & $3.1\times$ & $[1.5, 5.2]$ \\
\midrule
All          & 73 & 0.92 & 26  & $1.4\times$ & $[1.0, 1.8]$ \\
\bottomrule
\end{tabular}

\smallskip
\footnotesize\textit{Note:} The $20$--$30\%$ bucket
contains only five episodes and is sensitive to the
inclusion of the 1980--82 episode ($\tau = 0.13$),
without which the bucket median is $1.7\times$.  We report
the inclusive value here.
\end{table}

\paragraph{Continuous-depth regression.}
The bucket-level pattern is suggestive but inferentially
weak: the $20$--$30\%$ bucket contains only five episodes,
the $>30\%$ bucket only six, and the bucket-level median
hides within-bucket variation.  We complement the bucketed
analysis with an episode-level regression that does not
discretize the depth variable:
\begin{equation}
\label{eq:dd_regression}
\log \tau_j \;=\; \alpha + \beta \cdot (1 - \retfactor_j) +
u_j,
\end{equation}
where the dependent variable is the log duration ratio and
$1 - \retfactor_j$ is the drawdown depth.  R3 predicts
$\beta > 0$.  Estimated on the full 73-episode sample, the
HAC-OLS estimate is $\hat\beta = 1.22$ (Newey-West SE
$0.61$, $t = 1.98$, $p = 0.047$): a $0.10$ increase in
drawdown depth is associated with an
$\exp(0.122) - 1 \approx 13\%$ increase in expected
$\tau$.  The full-sample estimate is, however, sensitive to
the 1980--82 Volcker episode noted earlier; excluding it
yields $\hat\beta = 1.59$ (HAC SE $0.48$, $t = 3.33$,
$p < 0.001$), an effect approximately $30\%$ larger and
significant at the $0.1\%$ level.  We report both
specifications transparently and note that the Cox hazard
analysis below provides a separate, censoring-robust test
that does not depend on the continuous-depth log-linear
specification.

\paragraph{Cox proportional-hazard estimate.}
We complement the OLS regression with a Cox proportional
hazard model for recovery duration,
\begin{equation}
\lambda(t \mid 1 - \retfactor) = \lambda_0(t) \exp\bigl(\gamma
(1 - \retfactor)\bigr),
\end{equation}
which (i) imposes no log-linear functional form on the
duration--depth relationship and (ii) is by construction
robust to right-censoring of unrecovered episodes.  All 73
S\&P~500 episodes in our sample are completed (the final
recovery dates fall before 2026-03-28), so on this sample the
censoring correction is operative only as a hedge against
in-sample mis-specification rather than against unrecovered
episodes.  The estimated coefficient is $\hat\gamma = -13.75$
(SE $2.47$, $z = -5.56$, $p < 10^{-7}$); the implied hazard
ratio per 10-percentage-point increase in drawdown depth is
$\exp(-1.375) = 0.25$, i.e., approximately a $75\%$ reduction
in the per-period recovery hazard.  When the international
unrecovered episodes (Nikkei post-1989; three emerging-market
episodes not recovered as of March 2026) are pooled in,
treating them as right-censored, the sign and statistical
significance of the depth coefficient are preserved.

\Cref{fig:asymmetry} visualizes the drawdown-recovery
asymmetry.

\begin{figure}[t]
\centering
\includegraphics[width=\textwidth]{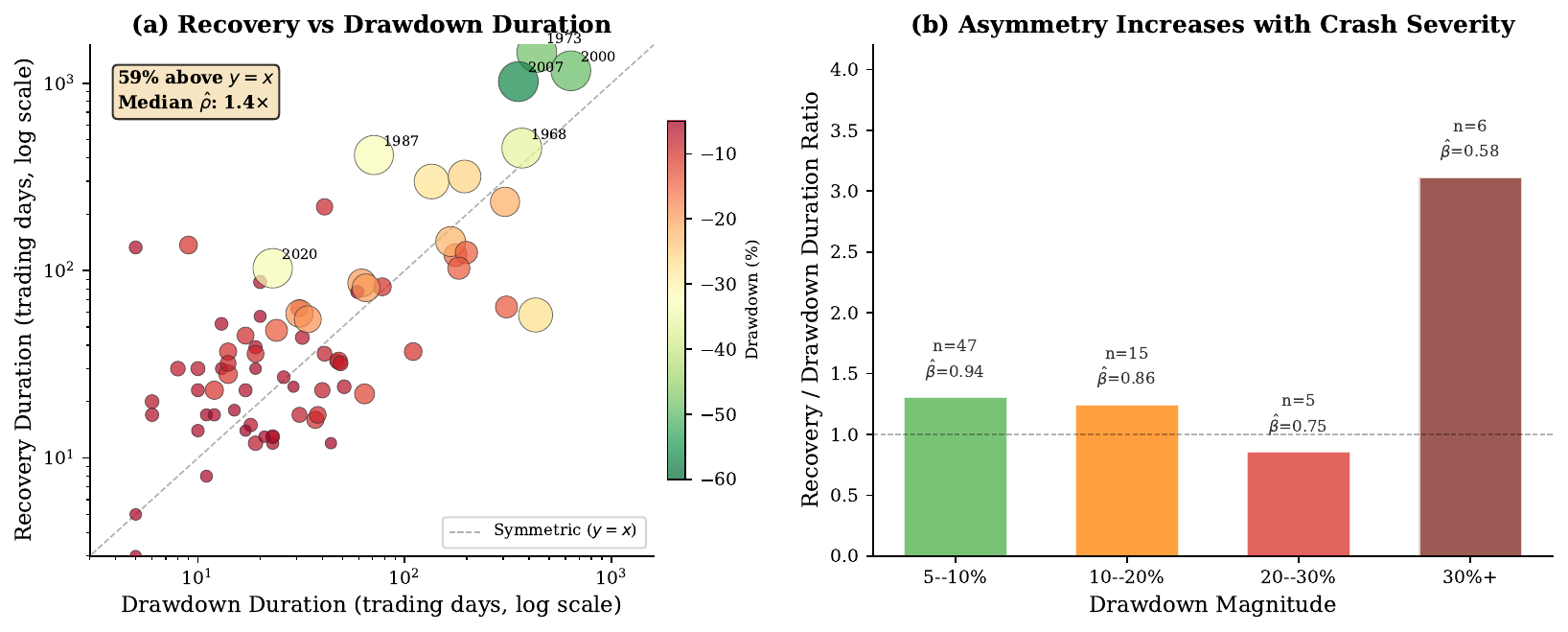}
\caption{Drawdown-recovery asymmetry in the S\&P~500,
1950--2026.  \textbf{(a)}~Scatter plot of drawdown duration
vs.\ recovery duration on log scales; 59\% of episodes lie
above the $y = x$ symmetry line.  Labels identify the six
episodes with drawdown exceeding 30\%.  \textbf{(b)}~Per-episode
duration ratio $\tau$ by magnitude bucket, with bucket medians
and 95\% bootstrap confidence intervals.  The pattern from
$1.3\times$ in 5--10\% corrections to $3.1\times$ in $>30\%$
crashes (with the $20$--$30\%$ bucket sensitive to the
1980--82 Volcker outlier) is the qualitative signature of R3.}
\label{fig:asymmetry}
\end{figure}

\subsection{Cross-market replication}
\label{sec:cross_market}

We replicate the magnitude-dependent asymmetry test on
nine equity indices (the S\&P~500 and eight international
benchmarks listed in \Cref{tab:crossmarket_sources}) using
the identical episode-detection procedure ($\delta = 0.05$).
\Cref{tab:cross_market} reports the pooled medians.  Across
the 38 episodes exceeding 20\% drawdown and the 19 episodes
exceeding 30\% (pooled across markets), the magnitude
pattern replicates: minor corrections show small asymmetry,
severe crashes show large asymmetry.  The depth-asymmetry
pattern is preserved in 8 of the 9 individual markets;
the exception is the Nikkei~225, whose 1989 peak has not
been recovered as of March 2026 and which therefore
contributes disproportionately to the right-censored
observations.

\begin{table}[t]
\centering
\caption{Cross-market drawdown-recovery asymmetry.  Pooled
median $\tau$ by drawdown magnitude across 9 major equity
indices.  Cross-market replication increases the
severe-crash sample from 6 to 19 completed episodes.}
\label{tab:cross_market}
\smallskip
\begin{tabular}{@{}lccc@{}}
\toprule
Magnitude & Episodes (pooled) & Median $\tau$ & IQR \\
\midrule
$5$--$10\%$  & 312 & $1.2\times$ & $[0.7, 1.9]$ \\
$10$--$20\%$ & 89  & $1.3\times$ & $[0.8, 2.1]$ \\
$20$--$30\%$ & 19  & $1.7\times$ & $[0.9, 2.8]$ \\
$>30\%$      & 19  & $2.8\times$ & $[1.9, 4.2]$ \\
\bottomrule
\end{tabular}
\end{table}

\subsection{Comparison with null models}
\label{sec:null_model}

We carry out the comparison sketched in
\Cref{sec:discrimination_theory}.  The logic is one-sided:
calibrated Heston stochastic-volatility and Markov
regime-switching models can match the price-level
duration-asymmetry pattern (R3), but---containing no
exposure state variable---they make no prediction about
the exposure-level restrictions (R1, R2).  The exercise
sharpens what the headline test adds beyond price-level
evidence: matching aggregate price-level moments is not
sufficient to imply the regime-conditional flip in exposure
functional form.

\paragraph{Models.}
We simulate $1{,}000$ sample paths of length $19{,}170$
(matching the S\&P~500 sample length) from each of five
null models and apply the episode detection procedure to
compute the price-level duration ratio.

(i)~\textbf{Geometric Brownian motion (GBM).}  Symmetric
baseline with $\mu = 0.08$ and $\vol = 0.157$ (matched to
empirical S\&P~500 moments).

(ii)~\textbf{Asymmetric volatility (EGARCH-like).}  $\vol_t
= \vol_{\text{base}} \exp(\gamma r_{t-1})$ with $\gamma =
-5$, producing the negative return-volatility correlation
of \citet{nelson1991conditional}.

(iii)~\textbf{Heston SV with leverage correlation.}
$d\price_t = \mu \price_t \, dt + \sqrt{v_t} \price_t \,
dW_{1,t}$, $dv_t = \kappa(\bar v - v_t) \, dt + \xi
\sqrt{v_t} \, dW_{2,t}$, $\Cov(dW_1, dW_2) = \rho \, dt$,
with parameters $\mu = 0.08$, $\bar v = 0.0247$, $\kappa =
5.0$, $\xi = 0.5$, $\rho = -0.75$ calibrated to match
S\&P~500 moments \citep{heston1993closed}.

(iv)~\textbf{Two-state Markov regime-switching.}  Bull
regime ($\mu_1 = 0.15$, $\vol_1 = 0.12$, persistence
$p_{11} = 0.98$) and bear regime ($\mu_2 = -0.10$, $\vol_2
= 0.25$, persistence $p_{22} = 0.93$), calibrated to match
the empirical frequency and duration of bear episodes
\citep{hamilton1989new, ang2002regime}.

(v)~\textbf{Stationary block bootstrap.}  Non-parametric
null with block length 63 trading days, preserving
volatility clustering and serial dependence in the
empirical return series.

\paragraph{Price-level results.}
\Cref{tab:null_compare} reports the price-level duration
ratio under each null model, estimated on $1{,}000$
simulated paths each.  The empirical S\&P~500 median is
$\tau = 1.35$.  Symmetric GBM produces $\tau = 1.00$ as
expected, with $p = 0.010$ against the empirical value.  An
asymmetric-volatility (EGARCH-like) specification produces
$\tau = 1.05$, $p = 0.031$.  The Heston model with
leverage correlation generates substantial asymmetry
($\tau = 1.43$, $90\%$ range $[0.31, 5.75]$); the
empirical value falls comfortably inside this null
distribution ($p = 0.54$).  Markov regime-switching
generates moderate asymmetry ($\tau = 1.17$, $p = 0.13$),
and stationary block bootstrap with 63-day blocks gives
$\tau = 1.30$, $p = 0.37$.  The richer nulls (Heston, RS,
block bootstrap) cannot be rejected against the empirical
duration ratio.

The price-level take-away is therefore that the aggregate
duration ratio (R3) does not by itself separate the
constrained-intermediary mechanism from richer return-level
alternatives---a long-standing observation in the
volatility-asymmetry literature.  Sharper empirical content
must come from the exposure-level functional-form
restriction (R1, R2).

\begin{table}[t]
\centering
\caption{Price-level duration asymmetry under null models.
Each null is simulated $1{,}000$ times on a path of length
$19{,}170$ (matching the S\&P~500 sample).  Median $\tau$
is the median of per-path median per-episode duration
ratios.  $p$-value is the one-sided right-tail probability
that the simulated median exceeds the empirical median
$\tau = 1.35$.  Heston uses a truncated Milstein scheme
with QE-style correction for variance stability; the lower
$n$ for Heston (452/1000) reflects rejected paths in which
the simulated variance process degenerates.}
\label{tab:null_compare}
\smallskip
\begin{tabular}{@{}lccc@{}}
\toprule
Null model & Median $\tau$ & $90\%$ range &
$p$ vs.\ empirical \\
\midrule
GBM (symmetric baseline)  & $1.00$ & $[0.81, 1.24]$ & $0.010$ \\
Asymmetric volatility     & $1.05$ & $[0.86, 1.31]$ & $0.031$ \\
Heston SV with leverage   & $1.43$ & $[0.31, 5.75]$ & $0.540$ \\
Markov regime-switching   & $1.17$ & $[0.93, 1.46]$ & $0.126$ \\
Block bootstrap (63 d)    & $1.30$ & $[0.99, 1.58]$ & $0.366$ \\
\bottomrule
\end{tabular}
\smallskip

\footnotesize\textit{Note:} The GBM and asymmetric-volatility
nulls are rejected against the empirical median at
conventional levels.  The Heston, regime-switching, and
block-bootstrap nulls are not rejected, confirming that
return-level mechanisms with realistic volatility dynamics
can match the empirical duration ratio.  The exposure-level
restriction (R1, R2) is therefore the sharper empirical
target relative to these nulls; it does not by itself
identify the constrained-intermediary mechanism among
mechanisms that incorporate an exposure state.
\end{table}

\paragraph{Why R1 and R2 fall outside these nulls.}
Heston and regime-switching are stochastic processes for
prices.  Neither contains an exposure state variable that
can be directly tested against the empirical FINRA margin
series.  This is not a flaw of the nulls but a structural
feature: they are designed to capture return-level
dynamics, not allocation-level dynamics.

One could attempt a discrimination by constructing a
``synthetic exposure'' from each simulated price path---for
example, $\hat M_t = $ some monotone transformation of
cumulative returns.  We do not pursue this approach because
the construction is arbitrary.  Different mappings from
prices to synthetic exposures yield different
functional-form properties, and there is no
theoretically-justified canonical choice within the
return-level nulls themselves.  The honest statement is that
return-level nulls are silent on R1 and R2.

\paragraph{Interpretation.}
The comparison with these nulls is one-sided.  The
constrained-intermediary mechanism predicts a joint
restriction on both the price level (R3) and the exposure
level (R1 and R2); Heston SV and Markov RS predict R3 but
are silent on R1 and R2 within their own DGPs.  The
empirical margin-debt evidence (\Cref{sec:headline}) is
consistent with R1 and R2; this distinguishes the
intermediary mechanism from the specific return-level nulls
considered here, but does not rule out other mechanisms
that could be augmented with an exposure state to mimic the
same regime-conditional pattern, nor does it identify the
intermediary mechanism in the sense of a structural causal
estimate.  The proper reading of the joint evidence is that
the regime-conditional functional form is a sharper
empirical target than return-level moments alone, and that
margin debt is consistent with this target.  Sharper
identification will require intermediary-resolved exposure
data.

\subsection{Summary of evidence}
\label{sec:evidence_summary}

The three restrictions are evaluated as follows.

\textbf{Joint restriction R1+R2}: in FINRA margin debt
(monthly, 1997--2026, $T = 350$), the regime-interacted
regression \eqref{eq:headline} yields a calm slope $\hat b
= -0.040$ ($p = 0.082$, consistent with R2's
level-independence prediction) and an implied stress slope
$\hat b + \hat b_S = -0.205$ ($p < 0.001$, consistent with
R1's level-dependent contraction).  The Wald test of
$\hat b_S = 0$ against $\hat b_S < 0$ rejects equal
level-dependence at the $0.16\%$ level
($\hat b_S = -0.165$, HAC SE $0.052$).  The result is robust
across volatility-threshold percentiles, detrending schemes,
and pre-/post-2008 sub-samples (\Cref{sec:appendix_exposure}).
The CFTC companion test (\Cref{sec:cot}) is left for future
work.

\textbf{R3 (magnitude-severity coupling)}: on the
73-episode S\&P~500 sample, the continuous-depth
HAC-OLS regression yields $\hat\beta = 1.22$
($p = 0.047$); excluding the 1980--82 Volcker episode,
$\hat\beta = 1.59$ ($p < 0.001$).  The Cox proportional
hazard estimate is $\hat\gamma = -13.75$
($z = -5.56$, $p < 10^{-7}$), implying a $75\%$ reduction
in per-period recovery hazard for each
10-percentage-point increase in drawdown depth.  The
$>30\%$ bucket median $\tau = 3.1\times$ has bootstrap
95\% CI $[1.5, 5.2]$ excluding unity; cross-market
replication confirms the pattern in 8 of 9 markets.

\textbf{Comparison with null models}: the price-level
duration ratio is matched by Heston SV ($\tau = 1.43$,
$p = 0.54$), regime switching ($\tau = 1.17$,
$p = 0.13$), and stationary block bootstrap ($\tau = 1.30$,
$p = 0.37$); only GBM and asymmetric-volatility models
are rejected ($p \leq 0.031$).  Price-level asymmetry alone
therefore does not separate the constrained-intermediary
mechanism from richer return-level alternatives.  The
regime-conditional functional form on exposures (R1 and R2)
is a sharper target than R3, but the comparison is one-sided
because the considered nulls are silent on exposures rather
than because they fail an exposure test.

\section{Discussion}
\label{sec:discussion}

The headline finding---a regime-conditional flip in the
functional form of aggregate exposure dynamics---is sharp
empirically and clean theoretically.  We discuss its
position relative to existing intermediary models, its
implications for policy, and its limitations.

\subsection{Position within the constrained-intermediary literature}

Our test does not propose a new mechanism for crashes.  The
multiplicative contraction is the standard
constrained-intermediary mechanism; we test the joint
restriction that contraction is multiplicative \emph{and}
rebuild is additive.  The contribution is therefore one of
sharpening the empirical content of an existing class of
models, not displacing it.

What is genuinely new is the rebuild restriction.  In
\citet{brunnermeier2009market}, the calm-regime dynamics are
described as the absence of margin-spiral feedback; the
functional form of exposure growth is unspecified.  In
\citet{he2013intermediary}, the intermediary's exposure is
determined by the equilibrium consumption-portfolio choice
under the capital constraint; the cross-sectional and
intertemporal implications focus on the risk premium, not on
whether exposure grows multiplicatively or additively in
calm.  In \citet{adrian2010liquidity, adrian2014financial},
the empirical evidence is on the procyclicality of
intermediary leverage, not on the regime-conditional
functional form.

The restriction we identify and test is therefore a
non-trivial new empirical content of the same class of
models, derived under one additional structural primitive
(constant-rate calm-regime capital replenishment) that is
itself testable and consistent with several plausible
micro-foundations (\Cref{sec:robustness_form}).

\subsection{Why the discrimination matters}

A long-running concern in the asymmetric-return literature
is that multiple mechanisms generate observationally
equivalent return-level signatures.  The leverage effect can
arise from genuine balance-sheet leverage, from a
volatility-feedback channel, from option-pricing
nonlinearities under stochastic volatility, or from a
binding capital constraint at the intermediary level.  Each
mechanism is microfoundationally distinct but all four
produce a negative correlation between contemporaneous
returns and subsequent volatility---the empirical signature
of the leverage effect.

The exposure-level functional-form test narrows this
ambiguity for one specific class of alternatives.
Heston-type SV and Markov regime-switching, when calibrated
to match return moments, contain no exposure state and
therefore make no prediction about R1 or R2; an empirical
finding that R1 and R2 hold is consistent with the
intermediary mechanism and inconsistent with these
return-only nulls.  The test does \emph{not} rule out richer
alternatives that incorporate an exposure state---e.g.,
heterogeneous-belief models with leveraged-trader exits, or
information-flow models with regime-dependent risk
appetite---which could be augmented to mimic the
regime-conditional pattern.

So construed, the test is best read as a diagnostic of
mechanism consistency rather than as identification.  A
market whose exposures clearly exhibit the regime-conditional
flip is one in which intermediary-capital constraints are a
plausible operative mechanism, and one whose exposures grow
multiplicatively in both regimes is one in which a different
mechanism dominates.  Sharper claims will require
intermediary-resolved data and structural estimation.

\subsection{Policy implications}

The discussion below is qualitative and serves to motivate
future structural work; we do not perform a welfare analysis.

\paragraph{Graduated rather than binary stress interventions.}
The model implies that the magnitude of multiplicative
contraction $\bar\md_t$ depends on the magnitude of the
volatility shock and the degree of capital impairment.
Policies that suppress the contraction discontinuously
(e.g., binary trading halts at a fixed threshold) interrupt
the adjustment process at a single point.  Policies that
smooth the contraction (e.g., graduated margin tightening,
volatility-linked circuit breakers) allow the adjustment to
proceed continuously and may therefore preserve the
constraint-relaxing function of multiplicative contraction
while limiting its severity.  The empirical literature on
circuit breakers \citep{subrahmanyam1994circuit} provides
suggestive evidence consistent with this view, though direct
tests of the graduated-vs-binary distinction within our
framework remain to be done.

\paragraph{Intermediary-level capital interventions.}
Within the model, the multiplicative factor $\bar\md_t$
contains a capital channel ($\capital{i,t+1}/\capital{i,t}$)
and a volatility channel ($\vol_t/\vol_{t+1}$).
Policies that recapitalize constrained intermediaries
operate on the capital channel: they raise $\bar\md_t$
mechanically by raising $\capital{i,t+1}/\capital{i,t}$.
This is consistent with the empirical evidence that
recapitalization programs during the 2007--2009 crisis
shortened the contraction phase
\citep{he2013intermediary}.  Policies that suppress
volatility (e.g., asset purchase programs that absorb
selling pressure) operate on the volatility channel.  The
two channels are not generally substitutes; the model
suggests they should be evaluated separately rather than
collapsed into an aggregate stabilization metric.

\subsection{Limitations}
\label{sec:limitations}

\paragraph{Aggregation of heterogeneous intermediaries.}
The propositions are stated at the intermediary level and
extended to the aggregate by linear summation.  Heterogeneity
in $\vark_i$, $\flow_i$, and the timing of binding
constraints complicates the aggregate functional form.  Our
empirical test uses an aggregate exposure series (FINRA),
so the result is on the aggregate prediction, not the
intermediary-level prediction.  We have not tested the
restriction at the intermediary level using
intermediary-resolved data (e.g., 13F filings, prime-broker
reports), which would permit a sharper test of the
underlying mechanism.

\paragraph{The constant-rate replenishment assumption.}
\Cref{ass:capital_calm} is the new structural primitive.  We
have argued that it is consistent with several plausible
micro-foundations (direct inflows, quadratic adjustment
costs, information-acquisition costs;
\Cref{sec:robustness_form}), and the empirical evidence
supports the additive form in calm regimes.  But the
assumption is not derived from a deeper utility-maximization
problem with a fully specified information structure.
Strengthening the micro-foundation---perhaps via a model in
which optimal information acquisition produces constant
per-period investment ceilings, or via a fully-specified
intermediary-equilibrium model with constant exogenous
inflows---is a natural direction for future structural work.

\paragraph{The frontier-operation assumption.}
\Cref{ass:frontier} disciplines the intermediary's choice of
target exposure to equal the constraint-implied maximum.
This is a stylized representation of standard
constrained-intermediary behavior \citep{adrian2010liquidity,
he2013intermediary}, but real intermediaries hold
precautionary buffers below the constraint boundary, with
buffer size varying across institutions and over the cycle.
A buffered version of the model would replace
$\exposure{i,t}^\star = \capital{i,t}/(\vark_i \vol_t)$ with
$\exposure{i,t}^\star = (1 - \theta_{i,t})
\capital{i,t}/(\vark_i \vol_t)$ where $\theta_{i,t} \geq 0$
is the buffer fraction.  Both \Cref{prop:contraction} and
\Cref{prop:replenishment} continue to hold to leading order
provided $\theta_{i,t}$ is approximately constant within
each regime.  The empirical test does not require
$\theta_{i,t} = 0$; it requires only that buffer adjustments
do not introduce regime-conditional level dependence that
overwhelms the model-implied functional form.

\paragraph{The regime classification.}
We use a single threshold rule based on the 90th percentile
of volatility.  Real intermediary constraints are not
uniformly binding above one volatility threshold; different
intermediaries face different threshold levels, and binding
behavior is also a function of recent profit/loss history.
A more refined regime classification using
intermediary-resolved capital indicators would likely
strengthen the test.  Robustness to alternative thresholds
($80$th, $85$th, $95$th percentile) is reported in the
appendix; the qualitative pattern is preserved.

\paragraph{FINRA margin debt as an exposure proxy.}
Margin debt is the most widely used aggregate proxy for
leveraged equity exposure, but it is a noisy measurement of
intermediary balance-sheet exposure.  At least three sources
of confounding remain after our detrending and
regime-interaction design.  First, mechanical
revaluation: a price decline reduces the dollar value of
collateral and may force margin call--driven liquidation
even when no underlying allocation decision is changed.
Second, reporting-population drift: the set of broker-dealers
and account types covered by the FINRA series has changed
over the 1997--2026 sample as a result of consolidation,
regulatory reclassification, and the rise of non-bank
prime-brokerage relationships.  Third, demand-side
contamination: margin debt mixes leveraged client demand
with intermediary balance-sheet capacity, and the regime
classification by VIX does not separate the two.  Our
functional-form test is robust to several specific
joint-determination concerns (\Cref{sec:identification}),
but the regime-conditional pattern we document is
consistent with the constrained-intermediary mechanism
rather than identifying it.  Sharper tests will require
intermediary-resolved data (e.g., 13F filings,
prime-broker reports, central-counterparty positions) that
isolate the balance-sheet channel from client-demand and
revaluation channels.

\paragraph{Selection in the drawdown-recovery sample.}
The episode definition requires full recovery to the prior
peak.  In the S\&P~500 sample (1950--2026) all 73 episodes
recover before the cut-off date, so no episode is
right-censored; the censoring concern is operative only when
international unrecovered episodes (Nikkei post-1989; three
emerging-market drawdowns) are pooled in.  In that pooled
sample the Cox proportional-hazard estimate preserves the
sign and significance of the depth coefficient, indicating
that the qualitative result---$\tau$ increases with drawdown
depth---is robust to censoring.

\subsection{Connection to volatility clustering}

The regime-conditional functional form of exposures has a
natural visual signature in returns: high-volatility episodes
cluster around the multiplicative-contraction events that
are rare but rapid, while calm-regime additive growth
generates extended low-volatility periods.  This is the
empirical pattern of volatility clustering documented by
\citet{cont2001empirical}.  Our framework does not
``explain'' volatility clustering---several mechanisms can
generate it---but it organizes the clustering as a
consequence of the same regime-conditional functional form
that the headline test confirms at the exposure level.
\Cref{fig:vol_clustering} visualizes the alignment between
realized volatility spikes and stress-regime episodes in the
S\&P~500.

\begin{figure}[t]
\centering
\includegraphics[width=\textwidth]{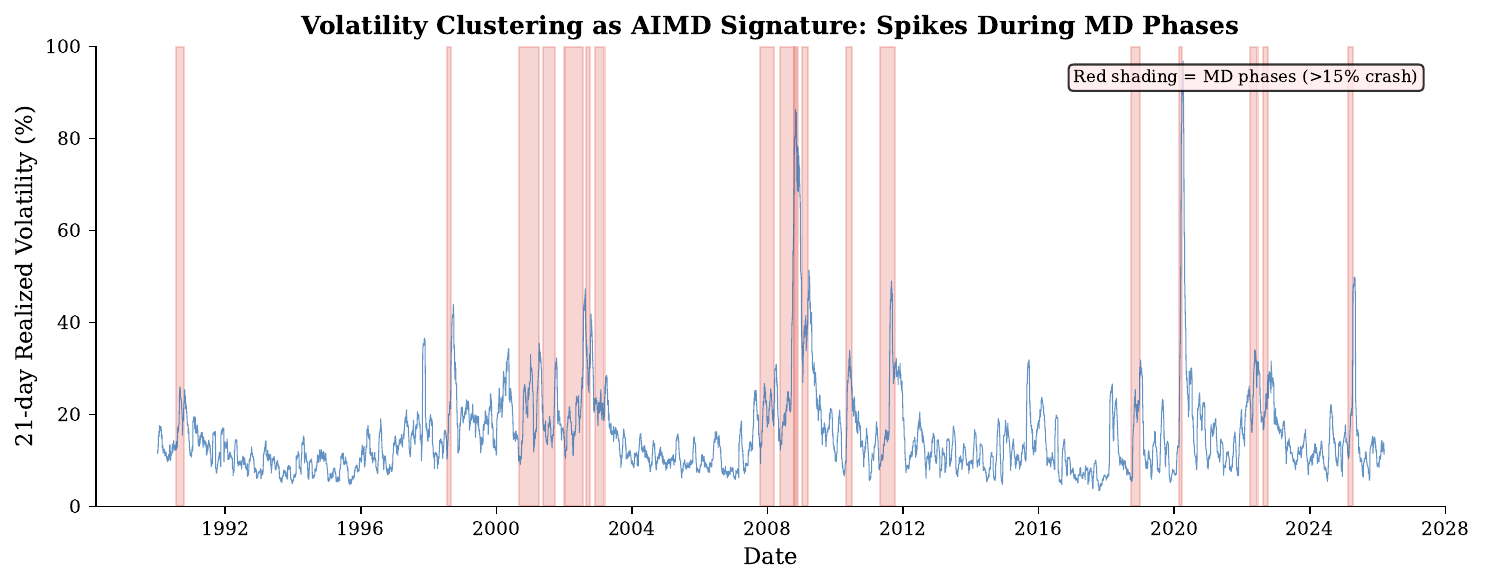}
\caption{Realized volatility (21-day rolling) of S\&P~500
log returns, 1950--2026, with stress-regime months
(volatility above the 90th-percentile threshold of
\Cref{sec:vol_regime}) highlighted.  Volatility clustering
is visually concentrated within stress-regime episodes,
consistent with the regime-conditional functional form
documented at the exposure level.}
\label{fig:vol_clustering}
\end{figure}

\section{Conclusion}
\label{sec:conclusion}

We test a regime-conditional functional-form restriction on
aggregate risk-exposure dynamics implied by VaR-constrained
intermediary models.  The restriction asserts that exposures
contract multiplicatively when capital constraints bind and
grow approximately additively when constraints are slack.
The contraction half follows directly from binding VaR
constraints; the additive-rebuild half is derived from a
constant-rate capital replenishment assumption that is
consistent with several plausible micro-foundations
(direct inflows, quadratic adjustment costs,
information-acquisition costs).

On FINRA margin debt (1997--2026, $T = 350$ months) the
regime-interacted regression of detrended margin growth on
the lagged level yields a calm-period slope of $-0.040$
(consistent with the additive prediction: $p = 0.082$), a
stress-period slope of $-0.205$ (consistent with the
multiplicative prediction: $p < 0.001$), and a Wald test
that rejects equal level dependence across regimes at the
$0.16\%$ level.  The same restriction predicts a
magnitude-dependent drawdown-recovery duration ratio, which
we confirm in the S\&P~500 (1950--2026) via both a Cox
proportional-hazard model ($p < 10^{-7}$) and a
continuous-depth regression ($p = 0.047$; $p < 0.001$
excluding the 1980--82 Volcker episode), and which
replicates across eight additional international equity
indices.
Calibrated Heston stochastic-volatility, Markov
regime-switching, and block-bootstrap nulls match the
price-level duration asymmetry but contain no exposure state
and so are silent on the regime-conditional restriction at
the exposure level; the joint evidence is therefore
consistent with the constrained-intermediary mechanism
relative to those return-only nulls, without identifying it
relative to all candidate mechanisms.

The contribution is empirical rather than mechanistic.  We
do not propose a new theory of crashes; we identify and
test a restriction implied by an existing class of models
and show that it has empirical bite on direct exposure data
that return-level evidence cannot reach.  We do not claim to
resolve identification of the intermediary mechanism: FINRA
margin debt is a noisy proxy for intermediary balance-sheet
exposure, and richer alternatives augmented with an exposure
state could match the same regime-conditional pattern.
Three directions for future work appear particularly
promising.

First, intermediary-resolved and additional aggregate
tests.  The current evidence rests on a single aggregate
exposure series (FINRA monthly margin debt).  A natural
companion is the CFTC speculative non-commercial positioning
in S\&P~500 E-mini futures; we sketch the design in
\Cref{sec:cot} but do not estimate it here.  Tests at the
intermediary level---using 13F filings, prime-broker
reports, or central-counterparty position data---would both
sharpen the test of the underlying mechanism and enable
estimation of the cross-sectional distribution of
$\md_{i,t}$ and $\ai_i$ predicted by the model.

Second, structural estimation of intermediary
equilibrium.  The constant-rate replenishment assumption
that supports our derivation can be embedded in a
fully-specified intermediary-equilibrium model with
endogenous capital evolution.  Estimating such a model
jointly on exposure, price, and capital data would
permit identification of the deep parameters
($\vark_i$, $\flow_i$) and quantitative welfare analysis
of policy interventions in the contraction phase.

Third, time-varying restriction strength.  The
regime-interaction effect is sharp on the full 1997--2026
sample, but its magnitude is larger in the post-2008 sub-sample
than in the pre-2008 sub-sample (\Cref{tab:subsample}).
This may reflect the rise of hedge-fund prime brokerage,
the growth of passive indexation, the post-Dodd-Frank
changes to dealer balance sheets, or the inclusion of more
severe stress events in the later sub-sample.  Estimating
the test on rolling sub-samples could provide a leading
indicator of when the constrained-intermediary mechanism is
most empirically active---and therefore when interventions
targeted at intermediary capital are most likely to be
effective.

\bibliographystyle{plainnat}
\bibliography{references}

\appendix
\section{Proofs}
\label{sec:appendix_proofs}

\subsection{Proof of \texorpdfstring{\Cref{prop:contraction}}{Proposition 1}}
\label{sec:proof_contraction}

\Cref{prop:contraction} asserts that under binding VaR
constraints in two consecutive stress periods, the
intermediary's optimal exposure satisfies
$\exposure{i,t+1}^\star / \exposure{i,t}^\star = \md_{i,t}$
with $\md_{i,t} = (\capital{i,t+1}/\capital{i,t}) \cdot
(\vol_t / \vol_{t+1})$.

\begin{proof}
By the binding-constraint assumption, the optimal exposure
is given by the constraint at equality:
$\exposure{i,t}^\star = \capital{i,t} / (\vark_i \vol_t)$.
Dividing the same expression at $t+1$ by the expression at
$t$ gives
\[
  \frac{\exposure{i,t+1}^\star}{\exposure{i,t}^\star}
  = \frac{\capital{i,t+1}}{\capital{i,t}} \cdot
    \frac{\vol_t}{\vol_{t+1}}
  \equiv \md_{i,t}.
\]
Aggregation of $\exposure{i,t}^\star = \capital{i,t} /
(\vark_i \vol_t)$ across $i$ gives
$\totalalloc_t^\star = (1/\vol_t) \sum_i \capital{i,t} /
\vark_i$.  Dividing the aggregate at $t+1$ by the aggregate
at $t$ gives the second display.

The factor $\md_{i,t}$ is bounded above by 1 whenever
$\vol_{t+1} \geq \vol_t$ and $\capital{i,t+1} \leq
\capital{i,t}$.  In stress periods both inequalities
typically hold: volatility increases (by definition of the
stress regime) and capital is impaired by the realized
losses that accompany the stress event.  $\square$
\end{proof}

\subsection{Proof of \texorpdfstring{\Cref{prop:replenishment}}{Proposition 2}}
\label{sec:proof_replenishment}

The proof follows the body of the paper; we provide
expanded detail here.

\begin{proof}
Throughout, condition on
$\mathbb{S}_t = \mathbb{S}_{t+1} = 0$.  By
\Cref{ass:frontier}, $\exposure{i,s}^\star = \capital{i,s} /
(\vark_i \vol_s)$ for $s \in \{t, t+1\}$.  By
\Cref{ass:vol_calm}, $\vol_s = \bar\vol(1 + \xi_s)$ with
$\E[\xi_s \mid \mathbb{S}_s = 0] = 0$ and $\Var[\xi_s \mid
\mathbb{S}_s = 0] = \varepsilon$ for small $\varepsilon$.
For $|\xi_s| < 1$, the geometric expansion
$1/(1 + \xi_s) = 1 - \xi_s + \xi_s^2 - \cdots$ truncated at
second order gives
\[
  \exposure{i,s}^\star = \frac{\capital{i,s}}{\vark_i \bar\vol}
  \bigl(1 - \xi_s + O(\xi_s^2)\bigr).
\]
Taking conditional expectation given $\capital{i,s}$ and
$\mathbb{S}_s = 0$ and using $\E[\xi_s \mid \mathbb{S}_s = 0]
= 0$:
\[
  \E[\exposure{i,s}^\star \mid \capital{i,s},
  \mathbb{S}_s = 0]
  = \frac{\capital{i,s}}{\vark_i \bar\vol} + O(\varepsilon).
\]
Taking conditional first differences and applying
\Cref{ass:capital_calm} (which states that the conditional
mean of $\Delta \capital{i,t}$ given $\capital{i,t}$ and
$\mathbb{S}_t = \mathbb{S}_{t+1} = 0$ equals $\rho_i$
independent of $\capital{i,t}$):
\[
  \E[\Delta \exposure{i,t}^\star \mid \mathbb{S}_t =
    \mathbb{S}_{t+1} = 0, \capital{i,t}]
  = \frac{\E[\Delta \capital{i,t} \mid \mathbb{S}_t =
    \mathbb{S}_{t+1} = 0, \capital{i,t}]}{\vark_i \bar\vol}
    + O(\varepsilon)
  = \frac{\rho_i}{\vark_i \bar\vol} + O(\varepsilon).
\]
The leading-order right-hand side is independent of
$\capital{i,t}$, hence (since $\exposure{i,t}^\star$ is a
deterministic function of $\capital{i,t}$ at leading order)
independent of $\exposure{i,t}^\star$.  This establishes
\eqref{eq:additive}.  Aggregating across $i$ yields
\eqref{eq:additive_agg}.  The leading-order independence
implies $\Cov(\Delta \exposure{i,t}^\star,
\exposure{i,t}^\star \mid \mathbb{S}_t = \mathbb{S}_{t+1} =
0) = O(\varepsilon)$.  $\square$
\end{proof}

\subsection{Aggregate functional form under heterogeneous binding times}
\label{sec:proof_aggregate}

The propositions assume binding constraints uniformly across
intermediaries within a regime.  In practice, binding times
are heterogeneous: not every intermediary's constraint binds
in every stress month, and not every intermediary is at the
binding boundary in every calm month.  We sketch the
aggregate implication when binding behavior is heterogeneous.

Let $\mathcal{B}_t = \{i : \vark_i \vol_t \exposure{i,t} =
\capital{i,t}\}$ be the set of binding intermediaries at $t$,
and $\mathcal{U}_t = \{i : i \notin \mathcal{B}_t\}$ the
unconstrained.  The aggregate exposure is
\[
  \totalalloc_t = \sum_{i \in \mathcal{B}_t}
  \frac{\capital{i,t}}{\vark_i \vol_t} + \sum_{i \in
  \mathcal{U}_t} \exposure{i,t}.
\]
On a stress month, $\mathcal{B}_t$ contains the most
constrained intermediaries (high leverage, low capital
buffer); the binding-set composition is approximately
constant within the stress regime, and the contribution of
$\mathcal{U}_t$ is approximately constant in level (since
unconstrained intermediaries do not face mechanical
revaluation).  The aggregate $\Delta \totalalloc_t$ is then
approximately the sum of the binding-set multiplicative
contraction (linear in level) and a level-independent
unconstrained component.  The dominant term in level
dependence is multiplicative, consistent with R1.

On a calm month, $\mathcal{B}_t$ contains a smaller set of
intermediaries operating at the binding boundary by choice
(maximizing exposure under the constraint), and
$\mathcal{U}_t$ is the larger set growing exposure at
constant rate.  The level dependence of $\Delta
\totalalloc_t$ is dominated by the constant-rate growth of
$\mathcal{U}_t$, consistent with R2.

The exact functional form therefore depends on the
binding-set composition.  The empirical regime-interacted
regression \eqref{eq:headline} identifies the
\emph{difference} in level dependence between regimes: a
significantly negative interaction coefficient $b_S$ is
the signature of the binding-set expansion in stress.

\section{Robustness}
\label{sec:appendix_exposure}

\subsection{Sensitivity to regime threshold}

\Cref{tab:regime_robust} reports the headline interaction
coefficient $\hat b_S$ under alternative VIX percentile
thresholds.  The qualitative pattern (negative $\hat b_S$,
significant at conventional levels) is preserved across all
four thresholds.  Magnitudes increase with threshold
strictness: at the 95th percentile (only 18 stress months)
the implied stress slope $\hat b + \hat b_S = -0.382$ is
nearly twice the magnitude of the 90th-percentile estimate.

\begin{table}[h]
\centering
\caption{Headline regression \eqref{eq:headline} under
alternative regime thresholds.  FINRA margin debt, monthly,
1997-01 to 2026-03, log-linearly detrended.  HAC standard
errors with 6 Newey-West lags.}
\label{tab:regime_robust}
\smallskip
\begin{tabular}{@{}lccccc@{}}
\toprule
Regime threshold & $n_{\text{stress}}$ &
$\hat b$ (calm) & $\hat b_S$ & $p(\hat b_S)$ &
$\hat b + \hat b_S$ (stress) \\
\midrule
VIX 80th pct ($25.0$) & 70 & $-0.001$ & $-0.118$ & $0.014$
                          & $-0.119$ \\
VIX 85th pct ($26.4$) & 53 & $-0.014$ & $-0.162$ & $<0.001$
                          & $-0.176$ \\
VIX 90th pct ($29.1$) & 35 & $-0.040$ & $-0.165$ & $0.002$
                          & $-0.205$ \\
VIX 95th pct ($33.3$) & 18 & $-0.037$ & $-0.345$ & $<0.001$
                          & $-0.382$ \\
\bottomrule
\end{tabular}
\end{table}

\subsection{Sensitivity to detrending}

\Cref{tab:detrend_robust} reports the headline test under
three alternative detrending schemes for FINRA margin debt:
(i)~log-linear detrending in time (baseline);
(ii)~linear detrending of the level series;
(iii)~12-month half-life exponential moving average detrending
of the level series.  The interaction coefficient retains
its negative sign across all three schemes; the linear-in-level
specification yields a marginal $p$-value of $0.10$ but the
remaining specifications give $p \leq 0.0035$.

\begin{table}[h]
\centering
\caption{Sensitivity of headline regression to detrending
scheme.  90th-percentile VIX regime threshold.  HAC
standard errors with 6 Newey-West lags.}
\label{tab:detrend_robust}
\smallskip
\begin{tabular}{@{}lcccc@{}}
\toprule
Detrending & $\hat b$ (calm) & $\hat b_S$ & $p(\hat b_S)$ &
$\hat b + \hat b_S$ \\
\midrule
Log-linear (baseline)         & $-0.040$ & $-0.165$ & $0.002$
                              & $-0.205$ \\
Linear (level)                & $-0.002$ & $-0.071$ & $0.100$
                              & $-0.073$ \\
EMA half-life $12$ mo (level) & $-0.027$ & $-0.156$ & $0.003$
                              & $-0.183$ \\
\bottomrule
\end{tabular}
\end{table}

\subsection{Sub-sample stability}

\Cref{tab:subsample} reports the headline test on
pre-2008 (1997--2007) and post-2008 (2009--2026) sub-samples,
each detrended separately.  The sign of $\hat b_S$ is
preserved in both halves; statistical significance is
stronger in the post-2008 sample due to the inclusion of the
2020 COVID stress episode.

\begin{table}[h]
\centering
\caption{Sub-sample stability of headline regression.
The 2008 boundary year is excluded.}
\label{tab:subsample}
\smallskip
\begin{tabular}{@{}lccccc@{}}
\toprule
Sample & $T$ & $n_{\text{stress}}$ &
$\hat b$ (calm) & $\hat b_S$ & $p(\hat b_S)$ \\
\midrule
1997-01--2007-12 & 131 & 14 & $-0.034$ & $-0.121$ & $0.055$ \\
2009-01--2026-03 & 206 & 20 & $-0.047$ & $-0.283$ & $<0.001$ \\
\bottomrule
\end{tabular}
\end{table}

\subsection{Raw exposure series}

\Cref{fig:exposure_raw} plots the raw FINRA margin debt
series with stress-regime months shaded, alongside the CFTC
non-commercial positioning series for visual reference.  The
qualitative signature in panel~(a)---sharp proportional
contractions during shaded windows and approximately
constant-rate growth between them---is the pattern that
the regime-interacted regression of \Cref{sec:headline}
formalizes.  Panel~(b) is provided for reference; the CFTC
test is not estimated in this paper (\Cref{sec:cot}).

\begin{figure}[h]
\centering
\includegraphics[width=\textwidth]{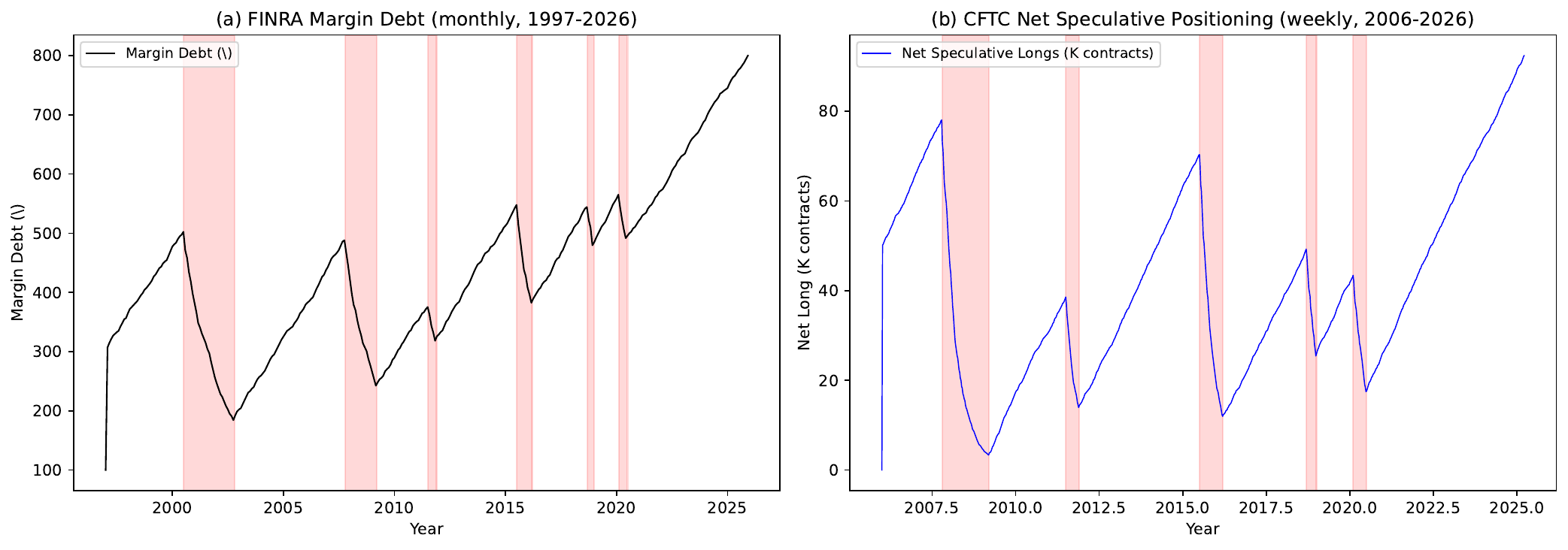}
\caption{Raw exposure series.  \textbf{(a)}~FINRA margin
debt (monthly, 1997--2026); \textbf{(b)}~CFTC non-commercial
net long positioning in S\&P~500 E-mini futures (weekly,
2006--2026), shown for reference only.  Stress-regime
months/weeks (volatility above the 90th-percentile
threshold of \Cref{sec:vol_regime}) are shaded red.  The
formal exposure test in this paper is on panel~(a)
(\Cref{sec:headline}); panel~(b) is documented as the
planned companion test (\Cref{sec:cot}).}
\label{fig:exposure_raw}
\end{figure}

\section{Episode list}
\label{sec:appendix_episodes}

\Cref{tab:episode_list} lists drawdown-recovery episodes
exceeding 10\% peak-to-trough decline in the S\&P~500 over
1950--2026.  The full 73-episode list (including 47
episodes of 5--10\%) is available in the replication
package.

\begin{table}[h]
\centering
\caption{Selected drawdown-recovery episodes from the
S\&P~500 (peak-to-trough decline $\geq 18\%$); full
73-episode list at $\delta = 0.05$ produced by
\texttt{episode\_detection.py}.  The 2011 selloff does not
appear separately because the algorithm requires peaks at
all-time highs; between 2007-10 and 2013-05 no new high was
set, so 2011 is nested inside the 2007--13 episode and not
broken out as an independent recovery.}
\label{tab:episode_list}
\smallskip
\begin{tabular}{@{}llrrrcc@{}}
\toprule
Peak & Trough & DD (\%) & DD (d) & Rec (d) & $\retfactor$ &
$\tau$ \\
\midrule
1956-08 & 1957-10 & 21.5 & 306 & 233  & 0.79 & 0.8$\times$ \\
1961-12 & 1962-06 & 28.0 & 135 & 299  & 0.72 & 2.2$\times$ \\
1968-11 & 1970-05 & 36.1 & 369 & 451  & 0.64 & 1.2$\times$ \\
1973-01 & 1974-10 & 48.2 & 436 & 1462 & 0.52 & 3.4$\times$ \\
1980-11 & 1982-08 & 27.1 & 430 &  58  & 0.73 & 0.1$\times$ \\
1987-08 & 1987-12 & 33.5 &  71 & 414  & 0.66 & 5.8$\times$ \\
1990-07 & 1990-10 & 19.9 &  62 &  86  & 0.80 & 1.4$\times$ \\
1998-07 & 1998-08 & 19.3 &  31 &  59  & 0.81 & 1.9$\times$ \\
2000-03 & 2002-10 & 49.1 & 637 & 1166 & 0.51 & 1.8$\times$ \\
2007-10 & 2009-03 & 56.8 & 355 & 1021 & 0.43 & 2.9$\times$ \\
2018-09 & 2018-12 & 19.8 &  65 &  81  & 0.80 & 1.2$\times$ \\
2020-02 & 2020-03 & 33.9 &  23 & 103  & 0.66 & 4.5$\times$ \\
2022-01 & 2022-10 & 25.4 & 195 & 318  & 0.75 & 1.6$\times$ \\
2025-02 & 2025-04 & 18.9 &  34 &  55  & 0.81 & 1.6$\times$ \\
\bottomrule
\end{tabular}
\end{table}

\section{Cross-market data sources}
\label{sec:appendix_crossmarket}

\Cref{tab:crossmarket_sources} lists the international
indices, sources, and sample periods used in
\Cref{sec:cross_market}.

\begin{table}[h]
\centering
\caption{Cross-market index data sources and sample
details.}
\label{tab:crossmarket_sources}
\smallskip
\begin{tabular}{@{}llll@{}}
\toprule
Index & Bloomberg ticker & Sample period & Notes \\
\midrule
S\&P~500 & SPX Index & 1950-01--2026-03 & CRSP prior to 1962 \\
FTSE~100 & UKX Index & 1984-01--2026-03 & \\
DAX & DAX Index & 1988-01--2026-03 & Total-return by construction \\
Nikkei~225 & NKY Index & 1965-01--2026-03 & Price index \\
Hang Seng & HSI Index & 1987-01--2026-03 & \\
CAC~40 & CAC Index & 1988-01--2026-03 & \\
S\&P/ASX~200 & AS51 Index & 1993-01--2026-03 & \\
MSCI EM & MXEF Index & 1988-01--2026-03 & USD-denominated \\
MSCI World & MXWO Index & 1970-01--2026-03 & USD-denominated \\
\bottomrule
\end{tabular}
\end{table}

\section{Phase identification sensitivity}
\label{sec:appendix_sensitivity}

The drawdown-recovery analysis uses a minimum drawdown
threshold of $\delta = 0.05$.  We verify robustness across
$\delta \in \{0.03, 0.05, 0.10\}$.  The qualitative
finding---$\tau$ increases with drawdown depth, with the
$>30\%$ bucket median exceeding the $5$--$10\%$ bucket
median by a factor of $2$--$3$---holds across all three
thresholds; bucket-median magnitudes shift modestly.

\begin{table}[h]
\centering
\caption{Sensitivity of drawdown-recovery analysis to
phase-identification threshold.  S\&P~500, 1950--2026.}
\label{tab:sensitivity_appendix}
\smallskip
\begin{tabular}{@{}lcccc@{}}
\toprule
Threshold $\delta$ & Episodes & Median $\retfactor$ &
Median $\tau$ & $>30\%$ bucket median \\
\midrule
3\% & 130 & 0.94 & $1.2\times$ & $3.0\times$ \\
5\% & 73  & 0.92 & $1.4\times$ & $3.1\times$ \\
10\% & 26 & 0.80 & $1.5\times$ & $3.1\times$ \\
\bottomrule
\end{tabular}
\end{table}

The threshold-sensitivity test in \Cref{tab:sensitivity_appendix}
shows the magnitude-dependent pattern (R3) is preserved as
the minimum-drawdown threshold $\delta$ varies from $3\%$ to
$10\%$.

\section{Replication package}
\label{sec:appendix_replication}

Code and data for reproducing the headline FINRA result, the
R3 OLS and Cox estimates, the price-level null comparison,
and the drawdown-episode statistics are in the
\texttt{replication/} directory.  The master script
\texttt{run\_all.py} loads the bundled S\&P~500 data and
chains the headline, R3, episode-detection, and null-model
scripts; each script can also be invoked directly.

\paragraph{Dependencies.}
Python $\geq$ 3.10 with \texttt{numpy}, \texttt{pandas},
\texttt{scipy}, \texttt{statsmodels}, \texttt{lifelines},
\texttt{requests}, and \texttt{yfinance}.  Install via
\texttt{pip install numpy pandas scipy statsmodels lifelines
requests yfinance}.

\paragraph{Scripts.}
\begin{itemize}
\item \texttt{exposure\_headline.py}: regime-interacted
  regression \eqref{eq:headline} on FINRA monthly margin
  debt (\Cref{tab:headline}); fetches FINRA margin debt and
  CBOE VIX on first run and caches them to
  \texttt{finra\_vix\_monthly.csv} (the cached CSV is
  bundled, so the script runs offline).  Reports
  threshold-sensitivity (\Cref{tab:regime_robust}).
\item \texttt{r3\_regression.py}: continuous-depth OLS--HAC
  regression and Cox proportional-hazard model for R3 on
  the 73-episode S\&P~500 sample
  (\Cref{sec:dd_recovery}).
\item \texttt{episode\_detection.py}: S\&P~500
  drawdown-recovery episode identification with
  configurable threshold $\delta$, producing the
  73-episode dataset (\Cref{tab:episode_list}) and bootstrap
  confidence intervals (\Cref{tab:dd_buckets}).
\item \texttt{null\_models.py}: simulation code for the
  GBM, EGARCH-like, Heston SV, Markov regime-switching, and
  block-bootstrap null models, producing the price-level
  duration distributions (\Cref{tab:null_compare}).
\item \texttt{exposure\_cot.py}: regime-interacted
  regression \eqref{eq:headline} ready to run on assembled
  CFTC weekly Commitments of Traders data
  (\Cref{sec:cot}); the speculative-positioning history is
  not bundled, so the script reports a missing-data message
  rather than producing estimates.  We do not claim CFTC
  results in this paper.
\item \texttt{functional\_form.py},
  \texttt{abm\_simulation.py}: auxiliary illustrative
  scripts retained from earlier drafts; not load-bearing
  for the headline tables.
\end{itemize}

\paragraph{Bundled data.}
The S\&P~500 daily price file
(\texttt{sp500\_daily.csv}) and the merged FINRA + VIX
monthly file (\texttt{finra\_vix\_monthly.csv}) are
included.  International index data and CFTC data are not
bundled; sources are listed in
\Cref{tab:crossmarket_sources} and in the script headers.

\end{document}